\documentclass[lettersize,journal]{IEEEtran}
\usepackage{amsmath,amsfonts,amsthm}
\usepackage{algorithmic}
\usepackage[ruled,linesnumbered]{algorithm2e}

\usepackage{array}
\usepackage[caption=false,font=scriptsize,labelfont={sf, scriptsize},textfont=sf]{subfig}
\usepackage{textcomp}
\usepackage{stfloats}
\usepackage{url}
\usepackage{verbatim}
\usepackage{graphicx}
\usepackage{mdframed}

\usepackage{cite}
\usepackage{xcolor}
\usepackage{etoolbox}
\usepackage{xcolor}

\makeatletter
\newcommand{\removelatexerror}{\let\@latex@error\@gobble}
\makeatother

\newtheorem{lemma}{Lemma}
\newtheorem{definition}{Definition}
\newtheorem{theorem}{Theorem}

\begin{document}

\title{Efficient Byzantine-Robust Privacy-Preserving Federated Learning via Dimension Compression}

\author{Xian Qin, Xue Yang* and Xiaohu Tang, ~\IEEEmembership{Fellow,~IEEE}

\thanks{Xian Qin, Xue Yang and Xiaohu Tang are with the Information Coding and Transmission Key Laboratory of Sichuan Province, Southwest Jiaotong University, Chengdu 610032, China. (e-mail: \{xueyang, xhutang\}@swjtu.edu.cn, xq@my.swjtu.edu.cn.)}
\thanks{*Corresponding author: Xue~Yang (e-mail: xueyang@swjtu.edu.cn).}}

\markboth{Journal of \LaTeX\ Class Files,~Vol.~14, No.~8, August~2021}%
{Shell \MakeLowercase{\textit{et al.}}: A Sample Article Using IEEEtran.cls for IEEE Journals}

\maketitle
\begin{abstract}
Federated Learning (FL) allows collaborative model training across distributed clients without sharing raw data, thus preserving privacy. However, the system remains vulnerable to privacy leakage from gradient updates and Byzantine attacks from malicious clients. Existing solutions face a critical trade-off among privacy preservation, Byzantine robustness, and computational efficiency. We propose a novel scheme that effectively balances these competing objectives by integrating homomorphic encryption with dimension compression based on the Johnson-Lindenstrauss transformation. Our approach employs a dual-server architecture that enables secure Byzantine defense in the ciphertext domain while dramatically reducing computational overhead through gradient compression. The dimension compression technique preserves the geometric relationships necessary for Byzantine defence while reducing computation complexity from $O(dn)$ to $O(kn)$ cryptographic operations, where $k \ll d$. Extensive experiments across diverse datasets demonstrate that our approach maintains model accuracy comparable to non-private FL while effectively defending against Byzantine clients comprising up to $40\%$ of the network. 
Our approach also demonstrates substantial improvements in computational and communication efficiency. Experimental evaluation shows that the dimension compression technique achieves $25 \times \sim 35 \times$ reduction in computational overhead and $17 \times$ reduction in communication overhead compared to our non-compression version. When compared to state-of-the-art methods like ShieldFL \cite{ShieldFL}, our approach demonstrates order-of-magnitude improvements in both computational and communication efficiency while maintaining equivalent privacy guarantees and achieving superior Byzantine robustness comparable to FLTrust \cite{fltrust}. These substantial efficiency enhancements make secure FL practical for deployment in large-scale neural networks with millions of parameters.
\end{abstract}


\section{Introduction}
Federated Learning (FL) \cite{FedAvg} is a distributed machine learning approach that enables collaborative model training across multiple clients without direct data sharing, thereby preserving data privacy. In FL, clients compute local model updates using their private datasets and transmit these updates to a central server. The server aggregates these updates to refine a global model. This paradigm facilitates training on decentralized data sources, such as mobile or edge devices, without compromising sensitive information. Up to now, FL has gained considerable interest for applications in healthcare \cite{medicalsurvey}, finance\cite{FLSurvey-financial}, and the Internet of Things (IoT) \cite{IOT}.

However, FL presents inherent security and privacy challenges. Research \cite{DLG,privacyInference,Membershipinference} indicates that model updates can leak sensitive information about local data, even without direct data sharing. This is especially concerning when dealing with sensitive data in healthcare or finance. Furthermore, the presence of Byzantine clients introduces the risk of poisoning attacks \cite{local_poisoning,manipulating}, where malicious updates can compromise the integrity and reliability of the global model.

To address these privacy concerns, Privacy-Preserving Federated Learning (PPFL) techniques have been developed. These techniques leverage cryptographic primitives such as Homomorphic Encryption (HE) \cite{EMKSA} and Secure Multi-party Computation (SMC) \cite{LightSec} to safeguard local gradients. These methods ensure that individual updates remain confidential throughout the aggregation process, preventing adversaries from inferring sensitive client information. Simultaneously, Byzantine-robust FL methods mitigate the challenge of malicious clients by analyzing statistical properties or similarities between gradient updates to identify and mitigate poisoning attacks. For example, FLTrust \cite{fltrust} employs cosine similarity measurements with a server-side trusted gradient to identify and downweight potentially malicious updates, effectively preserving model integrity even when a subset of clients attempts to poison the training process.

While these individual approaches address either privacy or robustness concerns, several approaches have attempted to combine privacy-preserving techniques with Byzantine defense mechanisms to achieve both privacy and robustness simultaneously. However, these integrated methods often incur unaffordable computational overhead due to cryptographic operations performed in the ciphertext domain \cite{ShieldFL}. Because Byzantine defense algorithms typically require secure computation of vector operations, such as similarity measurements or comparison operations, which become computationally expensive when executed on encrypted data. Consequently, many existing solutions, such as \cite{SAMFL, RFed}, compromise robustness for efficiency by employing simplified defense strategies that avoid computationally intensive operations like L2-norm calculations, resulting in weaker security guarantees against sophisticated attacks. Other approaches compromise the privacy of local gradients by adding the same mask to the entire gradient vector, which exposes the shifted distribution of the gradients to the server \cite{LSFL,SplitAgg}.

In conclusion, existing solutions often face a trade-off between computational efficiency and security guarantees, leading to either compromised privacy or weakened defense mechanisms.

To address the critical trade-off between privacy, robustness, and efficiency in FL, in this paper, we propose a novel scheme that simultaneously achieves Byzantine robustness and privacy preservation without sacrificing computational efficiency. Our approach uniquely combines dimension compression based on the Johnson-Lindenstrauss (JL) transformation with additive masking and homomorphic encryption, enabling secure Byzantine defense in the encrypted domain while reducing computation complexity. As a result, our scheme maintains Byzantine robustness comparable to plaintext FLTrust \cite{fltrust} and privacy guarantees comparable to state-of-the-art approaches like ShieldFL \cite{ShieldFL}. The main contributions of our work are summarized as follows:
\begin{itemize}
  \item We propose a novel scheme that achieves privacy-preserving Byzantine defense through strategic integration of additive masking and homomorphic encryption. 
    Unlike ShieldFL's two-trapdoor homomorphic encryption approach, our method employs lightweight additive masking for client-side gradient protection, eliminating the need for clients to perform complex cryptographic operations. This design significantly reduces computational and communication overhead on resource-constrained edge devices while maintaining strong privacy guarantees.
    On the server-side, by applying Paillier homomorphic encryption only to the Byzantine defense mechanism rather than the entire aggregation process, we significantly reduce computational and communication overhead.
    Our design further enhances efficiency by enabling complex cryptographic operations to occur offline since they are performed only on the masks, which remain independent of the local training process.
    This design maintains equivalent privacy guarantees to ShieldFL \cite{ShieldFL} for gradient confidentiality throughout both Byzantine defense and aggregation processes while substantially reducing both computational and communication overhead.

    \item For robust Byzantine defense, we implement a secure and efficient version of cosine similarity in the ciphertext domain based on FLTrust \cite{fltrust}, which provides stronger robustness against sophisticated attacks compared to ShieldFL \cite{ShieldFL}.
    To overcome the computation bottleneck of FLTrust \cite{fltrust} in the ciphertext domain, we introduce dimension compression based on the JL transformation. Our approach projects high-dimensional gradients into a substantially lower-dimensional space ($k \ll d$) while preserving critical geometric relationships with provable error bounds. This technique dramatically reduces cryptographic computation complexity from $O(dn)$ to $O(kn)$, enabling practical deployment for large-scale neural networks. Experiments demonstrate that our method maintains robust Byzantine defense capabilities even at extreme compression ratios (up to $10,000 \times$ dimension reduction), while preserving model accuracy by applying compression selectively—only during the defense phase rather than the aggregation process. This strategic application of compression makes privacy-preserving Byzantine-robust FL feasible for modern deep neural networks with millions of parameters, addressing a critical scalability challenge in secure FL.
    
    \item We provide formal privacy guarantees through rigorous security analysis ensuring client gradient confidentiality. Our experiments across diverse datasets and attack scenarios demonstrate model accuracy comparable to non-private FL while defending against Byzantine clients comprising up to 40\% of the network. Through both theoretical analysis and experimental evaluation, we show that our scheme significantly outperforms state-of-the-art Byzantine-robust PPFL methods in terms of efficiency. Our non-compression version reduces computational overhead by over $500 \times$ and communication overhead by $6 \times$ compared to ShieldFL. Our compression technique further reduces computational overhead by $25 \sim 35 \times$ and communication overhead by $17 \times$ compared to our non-compression version, while preserving defense capabilities even at extreme compression ratios of $0.0001$, making privacy-preserving Byzantine-robust FL practical for neural networks with millions of parameters.
\end{itemize}

\section{Related Work} \label{sec:related_work}

Federated Learning (FL) \cite{FedAvg} enables collaborative model training without sharing raw data, though research shows model updates can still leak sensitive information \cite{DLG,Membershipinference,privacyInference}. To address this, Privacy-Preserving FL (PPFL) employs cryptographic techniques like Homomorphic Encryption (HE) \cite{EMKSA} and Secure Multi-party Computation (SMC) \cite{LightSec} to secure the aggregation process.

Simultaneously, Byzantine-robust FL methods defend against malicious clients using statistical filtering \cite{trimmedmean} or similarity-based approaches \cite{krum}. For instance, Krum \cite{krum} selects updates based on Euclidean distance, while FLTrust \cite{fltrust} employs norm and cosine similarity with a trusted gradient to identify and downweight malicious contributions.

Combining privacy preservation with Byzantine robustness presents significant challenges, particularly when performing Byzantine defense in the ciphertext domain. Existing solutions typically face a three-way trade-off between privacy, robustness, and computational efficiency.

Several approaches prioritize security but suffer from computational inefficiency. PEFL \cite{PEFL} uses a dual-server architecture with Paillier encryption and Pearson correlation coefficients for Byzantine defense, but requires shared encryption keys, creating vulnerability to collusion. ShieldFL \cite{ShieldFL} addresses collusion concerns by combining two trapdoor homomorphic encryption schemes, but significantly increases computational overhead due to complex cryptographic operations.

Other methods compromise robustness for efficiency. Computing norm values—critical for Byzantine defense—requires expensive homomorphic multiplication in the ciphertext domain, which creates a significant computational bottleneck in the ciphertext domain. SAMFL \cite{SAMFL} avoids this by allowing clients to compute norms locally before sending them to the server, creating vulnerabilities where malicious participants can manipulate norm values to circumvent defense mechanisms. Similarly, DefendFL \cite{DefendFL} and RFed \cite{RFed} rely solely on cosine similarity and inner products respectively, making them vulnerable to model poisoning attacks that exploit the absence of norm-based constraints. Li et al. \cite{Li2024} attempt to reduce computational burden by randomly sampling gradient dimensions, but remain vulnerable to attacks targeting the excluded dimensions.

Some approaches sacrifice privacy for efficiency. LSFL \cite{LSFL} employs additive secret sharing with Krum \cite{krum} for Byzantine defense, but later analysis \cite{LSFL-erratum} revealed this approach exposes gradients with shifted distributions. Split Aggregation \cite{SplitAgg} faces similar privacy vulnerabilities, as also identified in PEFL \cite{PEFL-erratum}.

In summary, existing approaches typically compromise either privacy, robustness, or efficiency. Our proposed dimension compression technique offers a promising direction to overcome these limitations and achieve balanced performance across all three dimensions.

\section{Preliminaries} \label{sec:foundational_concepts}
This section introduces the foundational concepts of FL, cryptographic primitives, and the JL transformation employed in our scheme.

\subsection{Federated Learning}
FL enables multiple clients to collaboratively train a machine learning model without directly sharing their local data. In FL, each client $\mathcal{C}_i$ possesses a private dataset $\mathcal{D}_i$ and computes local gradients using the current global model parameters. The server then aggregates these local updates to refine the global model. The standard FedAvg algorithm \cite{FedAvg} operates as follows:
\begin{enumerate}
  \item Each client computes local gradient: $\boldsymbol{g}_i^t=\nabla \sum_{(x,y)\in \mathcal{D}_i}l(f(x,\boldsymbol{W}^t), y)$
  \item Server aggregates gradients: $\boldsymbol{\bar{g}}^t=\frac{1}{n}\sum_{i=1}^{n}\boldsymbol{g}_i^t$
  \item Clients update model: $\boldsymbol{W}^{t+1}=\boldsymbol{W}^{t}-\eta\cdot\boldsymbol{\bar{g}}^t$
\end{enumerate}
\subsection{Additive Homomorphic Encryption}
Homomorphic Encryption (HE) \cite{HE, HEsurvey} enables computations on encrypted data without decryption. In this paper, we use the Paillier cryptosystem \cite{Paillier}, which provides additive homomorphic properties. The scheme consists of:
\begin{itemize}
  \item $\textbf{KeyGen}(\kappa) \rightarrow (pk,sk)$: Given security parameter $\kappa$, generates public key $pk = (N, N+1)$ and private key $sk = \lambda$, where $N = p \cdot q$ is the product of two large primes and $\lambda = \text{lcm}(p-1, q-1)$.
  
  \item $\textbf{Enc}(m_i, pk)\rightarrow c_i$: Encrypts plaintext $m_i \in \mathbb{Z}_N$ as:
  $$c_i = (N+1)^{m_i} \cdot r^N \bmod N^2, r \in \mathbb{Z}_N^*.$$
  
  \item $\textbf{Dec}(c_i, sk)\rightarrow m_i$: Decrypts ciphertext $c_i$ as:
  $$m_i = L(c_i^\lambda \bmod N^2) \cdot \lambda^{-1} \bmod N$$ 
  where $L(u) = \frac{u-1}{N}$.
  
  \item $\textbf{Add}(\{c_i\}_{i \leq n},  pk) \rightarrow c_{add}$: Performs homomorphic addition:
  $$c_{add} = \prod_{i \leq n}{c_i} \bmod N^2 = \text{Enc}(\sum_{i \leq n}{m_i}, pk).$$
  
  \item $\textbf{ScalarMul}(c_i, k, pk) \rightarrow c'$: Performs homomorphic scalar multiplication:
  $$c'_i = c_i^k \bmod N^2 = \text{Enc}(k \cdot m_i, pk).$$
\end{itemize}

\subsection{Johnson-Lindenstrauss Transformation}
The JL transformation is a dimension-reduction technique that projects high-dimensional vectors into a significantly lower-dimensional space while probabilistically preserving pairwise Euclidean distances within a small distortion bound.
This random linear mapping allows for substantial dimensionality reduction with provable error bounds, making it particularly valuable for computational efficiency in large-scale machine learning applications.
This random linear mapping $f: \mathbb{R}^d \rightarrow \mathbb{R}^k$ (where $k \ll d$) is typically implemented using a random matrix $\boldsymbol{R} \in \mathbb{R}^{k \times d}$ with entries drawn from a Gaussian distribution or other suitable random distributions.
The theoretical foundation for this transformation is the Johnson-Lindenstrauss Lemma \cite{JL-lemma}, which states:
\begin{lemma}[Johnson-Lindenstrauss (JL) Lemma] \label{lem:JL}
Given two points $x, y \in \mathbb{R}^d$ and any $0 < \epsilon < 1$, if $f: \mathbb{R}^d \rightarrow \mathbb{R}^k$ is a random linear mapping with $k = O(\frac{\log(1/\delta)}{\epsilon^2})$, then with probability at least $1-\delta$:
$$(1-\epsilon)\Vert x - y \Vert^2 \leq \Vert f(x) - f(y) \Vert^2 \leq (1+\epsilon)\Vert x - y \Vert^2$$
\end{lemma}

In practice, the dimension compression works effectively when 
$$k \geq \frac{4+2\ln(1/\delta)}{\epsilon^2}$$
which is sufficient to ensure that, with probability at least $1-\delta$,
distances between $x$ and $y$ are preserved up to a factor of $(1 \pm \epsilon)$.
The JL transformation is widely used in dimension compression, compressed sensing, and machine learning to reduce the computation complexity of algorithms and improve the efficiency of data processing.

\section{Models and Design Goals} \label{sec:system_model}
In this section, we present the formal model for our Byzantine-robust PPFL system. We begin by introducing the system architecture and key components. Subsequently, we outline the threat model, which characterizes both the Byzantine attack capabilities and privacy vulnerabilities. Finally, we formalize our security objectives and design requirements that guide the development of our proposed scheme.
\subsection{System Model}
Similar to prior works \cite{LSFL, PEFL, RFed, ShieldFL}, we consider a dual-server architecture. The system consists of two non-colluding servers (denoted as $\mathcal{S}_0$ and $\mathcal{S}_1$) and a set of $n$ clients (denoted as $\mathcal{C} = \{\mathcal{C}_1, \mathcal{C}_2,..., \mathcal{C}_n\}$). Two servers collaboratively manage the training process in a privacy-preserving manner. Each client $\mathcal{C}_i \in \mathcal{C}$ possesses a private dataset $\mathcal{D}_i$, which contributes to training a global model with $d$ parameters. As shown in Fig \ref{fig:architecture}, these types of entities are responsible for the following operations:
\begin{itemize}
  \item Client: Each client $\mathcal{C}_i$ is responsible for training the model with its private dataset $\mathcal{D}_i$. Upon receiving the global gradient $\boldsymbol{g}_{\text{global}}$ from the server $\mathcal{S}_0$, the client updates its local model with $\boldsymbol{g}_{\text{global}}$ and performs local training on $\mathcal{D}_i$ to compute the local gradient $\boldsymbol{g}_i$. $\mathcal{C}_i$ then sends the local gradient $\boldsymbol{g}_i$ to the server for defence and aggregation.
  \item Server $\mathcal{S}_0$: $\mathcal{S}_0$ is responsible for aggregating the local updates from all clients with the generated aggregation weights from $\mathcal{S}_1$ to obtain the global gradient $\boldsymbol{g}_{\text{global}}$. After that, the server broadcasts $\boldsymbol{g}_{\text{global}}$ to the clients for the next iteration.
  \item Server $\mathcal{S}_1$: $\mathcal{S}_1$ is responsible for performing the Byzantine defence for the collected local gradients and deciding the aggregation weight for each update. $\mathcal{S}_1$ calculates these weights with the help of $\mathcal{S}_0$, where it can learn only the cosine similarity and norm of local gradients but cannot access any other information about the individual gradient. It then sends these weights to $\mathcal{S}_0$ for the aggregation process.
\end{itemize}

\subsection{Threat Model} \label{sec:threat_model}
In this work, we consider two types of entities: the dual servers ($\mathcal{S}_0$ and $\mathcal{S}_1$) and the clients. 
Our threat model considers the following types of participants and their behaviors:
\begin{itemize}
  \item \textbf{Semi-honest servers:} Both $\mathcal{S}_0$ and $\mathcal{S}_1$ correctly follow the protocol but may attempt to infer clients' private gradients out of curiosity. Following previous work \cite{LSFL, PEFL, RFed, ShieldFL}, we assume the two servers will not collude with each other, though either server may collude with semi-honest clients to infer other clients' local gradients. The communication channels between all participants are assumed to be secure against external adversaries who might intercept or inject messages.
  \item \textbf{Semi-honest clients:} These clients correctly follow the protocol but may attempt to infer other clients' local gradients out of curiosity. Semi-honest clients may collude with other semi-honest clients or with one of the semi-honest servers to infer others' local gradients. This inference is concerning because local gradients can be exploited for reconstruction attacks \cite{DLG}.
  \item \textbf{Malicious clients:} These clients actively seek to undermine the learning process by injecting harmful data through model poisoning or data poisoning attacks to compromise the accuracy and reliability of the global model. Malicious clients may upload carefully crafted gradients to influence the model's convergence direction. They may also collude with each other to amplify the impact of their attacks.
\end{itemize}

Note that clients can simultaneously be both semi-honest and malicious—that is, a client may correctly follow the protocol to avoid detection while attempting to both disrupt the learning process and/or infer private information from other clients.

\subsection{Design Goals} \label{sec:design_goals}
Based on our system and threat models, we aim to develop an efficient Byzantine-robust PPFL scheme that achieves the following objectives:
\begin{itemize}
  \item \textbf{Model Quality}: The scheme must maintain model accuracy comparable to standard FedAvg \cite{FedAvg} in benign settings. This ensures that our defensive mechanisms do not compromise the learning performance when no attacks are present.
  
  \item \textbf{Robustness}: The system must effectively defend against poisoning attacks by reliably detecting malicious clients and minimizing their influence on the global model aggregation process.
  
  \item \textbf{Privacy}: The scheme must provide strong privacy guarantees: clients' model updates should not be exposed to any potentially curious servers or other clients.
  
   \item \textbf{Efficiency}: The scheme must minimize both communication and computational overhead, specifically:
    \begin{itemize}
    \item Reduce client-server and inter-server communication costs
    \item Minimize computation burden on clients and servers
    \item Scale efficiently to large model architectures
    \end{itemize}
\end{itemize}

\section{Method} \label{sec:methodology}
In this section, we present our efficient Byzantine-robust PPFL scheme. 
Our approach employs a hybrid cryptographic design combining lightweight additive masking with homomorphic encryption in a dual-server architecture. For Byzantine defense, we implement FLTrust \cite{fltrust} in a privacy-preserving manner. To overcome the computation bottleneck of performing complex operations in the ciphertext domain, we introduce JL transformation-based dimension compression.

Our hybrid approach strategically distributes computation between additive secret sharing (for gradient aggregation) and homomorphic encryption (for Byzantine defense). This design significantly reduces efficiency challenges compared to schemes like ShieldFL \cite{ShieldFL}, while maintaining strong security guarantees. The dimension compression projects high-dimensional gradients into a lower-dimensional space before applying homomorphic encryption operations, as illustrated in Figure \ref{fig:architecture}. This technique dramatically reduces computational overhead while preserving the essential geometric relationships required for Byzantine defense. Since compression is applied exclusively to Byzantine defense calculations and not to gradient aggregation, it preserves model update integrity while significantly reducing computational costs. Algorithm \ref{alg:workflow} presents our proposed scheme, addressing all design goals outlined in Section \ref{sec:design_goals} by providing a practical solution that minimizes communication and computational overhead without compromising security or model performance.

Below, we present a detailed description of our proposed solution, including the system architecture and the step-by-step protocol workflow.

\begin{figure*}[!t]
  \centering
  \includegraphics[width=0.8\textwidth]{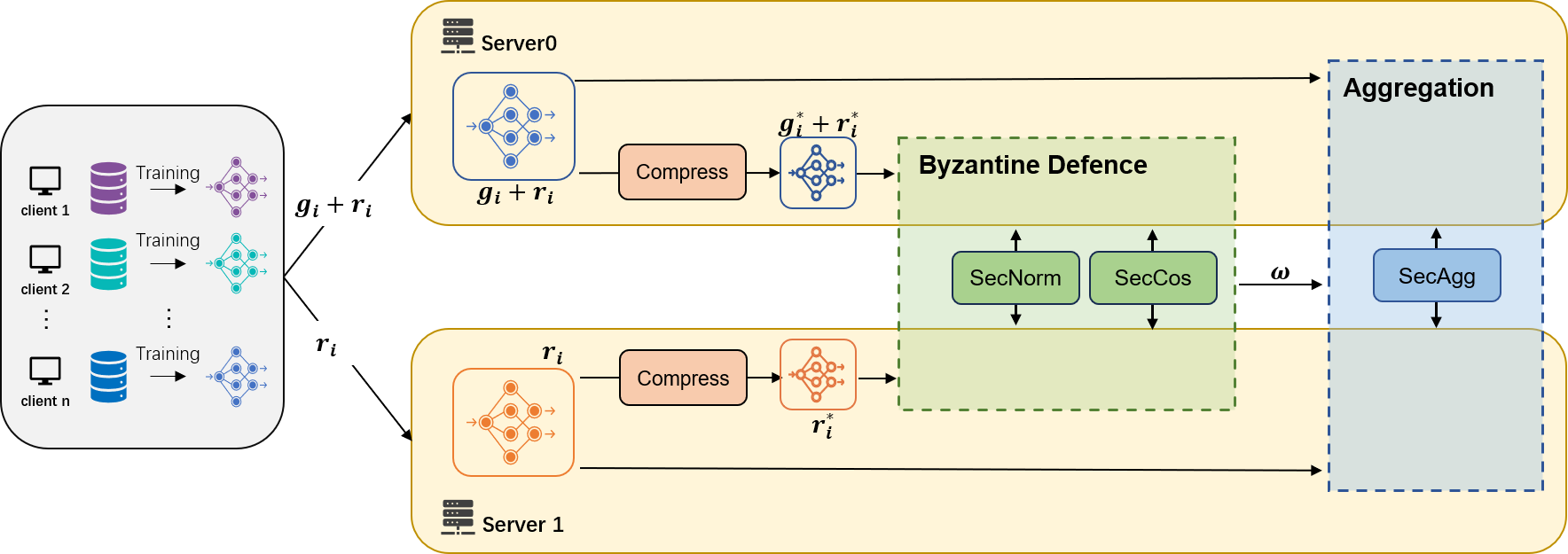}
  \caption{System Architecture of Byzantine-Robust Privacy-Preserving Federated Learning.}
  \label{fig:architecture}
  \end{figure*}

\subsection{Overview}
Our system architecture consists of two non-colluding servers ($\mathcal{S}_0$ and $\mathcal{S}_1$) and $n$ clients $\mathcal{C} = \{\mathcal{C}_1, \mathcal{C}_2,..., \mathcal{C}_n\}$, as illustrated in Fig. \ref{fig:architecture}. Each $\mathcal{C}_i$ holds a private dataset $\mathcal{D}_i$, while $\mathcal{S}_1$ maintains a small trusted dataset $\mathcal{D}_s$ to establish a baseline for Byzantine defence. 
Algorithm \ref{alg:workflow} outlines our protocol into five steps:
\begin{enumerate}
  \item \textbf{Initialization}: $\mathcal{S}_0$ generates a random projection matrix $R \in \mathbb{R}^{k \times d}$ and initializes the model parameters $\boldsymbol{W}^0$. $\mathcal{S}_1$ generates a Paillier key pair $(pk, sk)$ and sends $pk$ to $\mathcal{S}_0$. Each $\mathcal{C}_i$ downloads $\boldsymbol{W}^0$ from $\mathcal{S}_0$. To enable efficient communication during training, clients generate random seeds $s_i$ and share them with $\mathcal{S}_1$. Using these seeds, both $\mathcal{C}_i$ and $\mathcal{S}_1$ can independently generate identical random masks for each round without additional communication. Furthermore, $\mathcal{S}_1$ precomputes the compressed mask projections and their encrypted values offline before training begins. This offline preprocessing significantly reduces computational overhead during the actual training rounds, as these expensive cryptographic operations do not need to be repeated in each iteration.
  \item \textbf{Local Training}: In each training round $t$, each client $\mathcal{C}_i$ computes its local gradient $\boldsymbol{g}_i^t$ using the current model $\boldsymbol{W}^t$ and its private dataset $\mathcal{D}_i$. To protect privacy, each client using generated random mask $\boldsymbol{r}_i^t$ and sends the masked gradient $\boldsymbol{g}_i^t + \boldsymbol{r}_i^t$ to $\mathcal{S}_0$.
  \item \textbf{Poisoning Defence}: $\mathcal{S}_1$ computes a reference gradient $\boldsymbol{g}_{\text{standard}}^t$ using its trusted dataset $\mathcal{D}_s$. Subsequently, the two servers collaboratively perform secure computation of gradient norms $\Vert \boldsymbol{g}_i^t \Vert$ and cosine similarities $\text{cos}_i^t$ between each client gradient and the reference gradient using Algorithms \ref{alg:defence} (SecNorm) and \ref{alg:SecCos} (SecCos). These measurements enable detection of potentially malicious updates by identifying gradients with suspicious geometric properties while maintaining the confidentiality of individual client contributions.
  \item \textbf{Secure Aggregation}: Based on the calculated gradient norms $\Vert \boldsymbol{g}_i^t \Vert$ and cosine similarities $\text{cos}_i^t$, $\mathcal{S}_1$ assigns aggregation weights to each client update. Local gradients with higher similarity to the reference gradient $\boldsymbol{g}_{\text{standard}}^t$ receive larger weights, effectively reducing the influence of potentially malicious updates. The dual servers then collaboratively perform the weighted aggregation of local gradients in a privacy-preserving manner, ensuring that neither servers learn individual client contributions.
  \item \textbf{Update}: $\mathcal{S}_0$ broadcasts $\boldsymbol{g}_{\text{global}}^t$ to all clients, who then update their local models with the global gradient to obtain $\boldsymbol{W}^{t+1}$.
\end{enumerate}
Steps 2 to 5 are repeated for $T$ rounds until the global model converges. The detailed algorithms for each step are presented in the following sections.

\begin{algorithm}
  \caption{Byzantine-Robust Privacy-Preserving Federated Learning}\label{alg:workflow}
  \DontPrintSemicolon
  \KwIn{Security parameter $\kappa_1$; Clients $\mathcal{C} = \{\mathcal{C}_1, \mathcal{C}_2,..., \mathcal{C}_n\}$ with private datasets $\mathcal{D}_i$; Server $\mathcal{S}_1$ with standard dataset $\mathcal{D}_s$}
  \KwOut{Global model parameters $\boldsymbol{W}$}

  \tcc{Initialization}
  $\mathcal{S}_0$, $\mathcal{S}_1$, and all $\mathcal{C}_i \in \mathcal{C}$ collaboratively execute: $\textbf{Initialization}(\kappa_1, \kappa_2, k, d, T)$ 
  \tcp{Algorithm~\ref{alg:init}}

  \tcc{Local Training}
  \For{$t=0$ \KwTo $T$}{
    \For{$i=1$ \KwTo $n$ in parallel}{
      $\mathcal{C}_i$ computes local gradient $\boldsymbol{g}_i^t$ using $\boldsymbol{W}^t$ and $\mathcal{D}_i$\;
      
      $\mathcal{C}_i$ sends masked gradient $\boldsymbol{g}_i^t + \boldsymbol{r}_i^t$ to $\mathcal{S}_0$\;

    }

    \tcc{Poisoning Defence}
    $\mathcal{S}_1$ computes reference gradient $\boldsymbol{g}_{standard}^t$ using $\mathcal{D}_s$\;

    \For{$i=1$ \KwTo $n$}{
  
      $\mathcal{S}_0$ computes the compressed gradient:\\ 
       $\boldsymbol{g}_i^{t*} + \boldsymbol{r}_i^{t*} = \boldsymbol{R} \cdot (\boldsymbol{g}_i^t + \boldsymbol{r}_i^t)$\;
      $\Vert \boldsymbol{g}_i^t \Vert \approx \textbf{SecNorm}(\boldsymbol{g}_i^{t*} + \boldsymbol{r}_i^{t*}, \boldsymbol{c}_i^*)$
      \tcp{Algorithm \ref{alg:defence}} 

      $\text{cos}_i^t = \textbf{SecCos}(\boldsymbol{g}_i^t + \boldsymbol{r}_i^t, \boldsymbol{r}_i^t, \boldsymbol{g}_{standard}^t, \Vert \boldsymbol{g}_i^t \Vert)$ \tcp{Algorithm \ref{alg:SecCos}}
    }

    \tcc{Secure Aggregation}
    $\boldsymbol{g}_{\text{global}}^t = \textbf{SecAgg}(\{\boldsymbol{g}_i^t + \boldsymbol{r}_i^t\}, \{\boldsymbol{r}_i^t\}, \{\text{cos}_i^t\}, \{\Vert \boldsymbol{g}_i^t \Vert\})$ \tcp*{Algorithm \ref{alg:SecAgg}}

    \tcc{Update}
    $\mathcal{S}_0$ broadcasts $\boldsymbol{g}_{\text{global}}^t$ to all clients\;
    \For{$i=1$ \KwTo $n$ in parallel}{
      $\mathcal{C}_i$ updates local model: $\boldsymbol{W}_i^{t+1} = \boldsymbol{W}^{t} - \eta \cdot \boldsymbol{g}_{\text{global}}^t$\;
    }
  }
\end{algorithm}

\begin{algorithm}
  \caption{$\text{Initialization}(\kappa_1,\kappa_2,k,d,T)$}\label{alg:init}
  \DontPrintSemicolon
  \KwIn{Security parameters $\kappa_1$, $\kappa_2$; Projection dimension $k$; Model dimension $d$; Number of clients $n$; Number of rounds $T$}
  \KwOut{Initial model parameters $\boldsymbol{W}^0$; Compression matrix $\boldsymbol{R}$; Key pair $(pk, sk)$; Seeds $\{s_i\}_{i=1}^n$}

  $\mathcal{S}_0$ generates random projection matrix $\boldsymbol{R} \in \mathbb{R}^{k \times d}$ and initializes model parameters $\boldsymbol{W}^0$\;
  $\mathcal{S}_0$ generates modulus $q$ with $\vert q \vert = \kappa_2$\;
  $\mathcal{S}_0$ sends $\boldsymbol{R}$ to $\mathcal{S}_1$ and broadcasts $\boldsymbol{W}^0$ and $q$ to all clients\;

  $\mathcal{S}_1$ generates Paillier key pair $(pk, sk)$ with security parameter $\kappa_1$ and sends $pk$ to $\mathcal{S}_0$\;
  
  \tcc{Preprocessing for Efficiency}
  \For{$i=1$ \KwTo $n$ in parallel}{
    $\mathcal{C}_i$ generates random seed $s_i$ and send it to $\mathcal{S}_1$\;

  \For{$t=0$ \KwTo $T$}{

      $\mathcal{S}_1$  and $\mathcal{C}_i$ both generates identical mask using pseudorandom generation: $\boldsymbol{r}_i^t = G(s_i, t) \bmod q \in \mathbb{R}^d$\;

    $\mathcal{S}_1$ computes compressed mask $\boldsymbol{r}_i^{t*} = \boldsymbol{R} \cdot \boldsymbol{r}_i^t \in \mathbb{R}^k$\;
    $\mathcal{S}_1$ encrypts each component of the compressed mask:
      $\boldsymbol{c}^{t*}_i = [\textbf{Enc}(r^{t*}_{i,0}, pk), \ldots, \textbf{Enc}(r^{t*}_{i,k-1}, pk)]$\;
      $\mathcal{S}_1$ sends $\boldsymbol{c}^{t*}_i$ to $\mathcal{S}_0$ for later use\;

  }}
  \KwRet{$\boldsymbol{W}^0$, $\boldsymbol{R}$, $(pk, sk)$, $\{\boldsymbol{r}_i\}_{i=1}^n$, $\boldsymbol{c}^{t*}_i$}
\end{algorithm}

\subsection{Initialization}
The initialization phase prepares the system for secure and efficient FL, as shown in Algorithm \ref{alg:init}, which consists of system setup and computation preprocessing.

In system setup, $\mathcal{S}_0$ generates a random matrix $\boldsymbol{R} \in \mathbb{R}^{k \times d}$ (where $k \ll d$) for dimension compression and initializes model parameters $\boldsymbol{W}^0$. $\mathcal{S}_1$ generates a Paillier key pair $(pk, sk)$. Then, $\mathcal{S}_0$ sends $\boldsymbol{W}^0$ to all clients, and $\mathcal{S}_1$ shares $pk$ with $\mathcal{S}_0$.

For computational efficiency, we implement preprocessing (lines 5-12 in Algorithm \ref{alg:init}). Each client $\mathcal{C}_i$ generates a random seed $s_i$ and sends it to $\mathcal{S}_1$. Using these seeds, both $\mathcal{C}_i$ and $\mathcal{S}_1$ can generate identical masks $\boldsymbol{r}_i^t$ independently for each round $t$ using a secure pseudorandom number generator (PRNG) function $G(s_i, t) \bmod q$, where $q$ is a public modulus with bit length $\kappa_2$. 

This approach eliminates the need for additional communication during training.

With these masks, $\mathcal{S}_1$ precomputes compressed masks $\boldsymbol{r}_i^{t*} = \boldsymbol{R} \cdot \boldsymbol{r}_i^t$ and their encrypted forms $\boldsymbol{c}^{t*}_i$ before training begins. Since these computationally expensive operations are independent of training data, performing them offline before training begins significantly improves runtime efficiency.

\subsection{Local Training}
In each training round $t$, each client $\mathcal{C}_i$ locally computes its gradient $\boldsymbol{g}_i^t$ using the current model parameters $\boldsymbol{W}^t$ and its private dataset $\mathcal{D}_i$ by performing stochastic gradient descent (SGD):
$$\boldsymbol{g}_i^t=\nabla \sum_{(x,y)\in \mathcal{D}_i}l(f(x,\boldsymbol{W}^t), y)$$

To preserve privacy, each $\mathcal{C}_i$ masks its local gradient $\boldsymbol{g}_i^t$ with the random vector $\boldsymbol{r}_i^t$ generated in the initialization phase. $\mathcal{C}_i$ then sends only the masked gradient $\boldsymbol{g}_i^t + \boldsymbol{r}_i^t$ to $\mathcal{S}_0$, while $\mathcal{S}_1$ already possesses $\boldsymbol{r}_i^t$ from the preprocessing step.
This additive masking scheme creates a form of secret sharing where neither server can independently reconstruct the client's true gradient. By distributing knowledge of the gradient components ($\boldsymbol{g}_i^t + \boldsymbol{r}_i^t$ to $\mathcal{S}_0$ and $\boldsymbol{r}_i^t$ to $\mathcal{S}_1$), we achieve privacy guarantees provided the two servers remain non-colluding. This approach significantly reduces the computation burden on resource-constrained client devices compared to ShieldFL \cite{ShieldFL}, requiring only lightweight addition operations rather than computationally expensive cryptographic operations while maintaining strong privacy protection.

\subsection{Poisoning Defence}
To defend against Byzantine attacks, we implement a privacy-preserving version of FLTrust \cite{fltrust}, which uses cosine similarity and gradient normalization to identify and mitigate malicious updates. Our approach securely computes these metrics without exposing raw client gradients to either server.

The defense process consists of three key steps:
\begin{enumerate}
   \item \textbf{Reference Gradient Generation} (line 7): $\mathcal{S}_1$ computes a reference gradient $\boldsymbol{g}_{\text{standard}}^t$ using its trusted dataset $\mathcal{D}_s$. This reference serves as a baseline for detecting Byzantine behavior.
    
    \item \textbf{Dimension Compression} (line 10): To enable efficient computation while preserving Byzantine robustness, $\mathcal{S}_0$ applies dimension compression to the masked gradients using the projection matrix $\boldsymbol{R}$, computing $\boldsymbol{g}_i^{t*} + \boldsymbol{r}_i^{t*} = \boldsymbol{R} \cdot (\boldsymbol{g}_i^t + \boldsymbol{r}_i^t)$. 
      
      \item \textbf{Secure Metric Computation} (lines 11-12): $\mathcal{S}_0$ and $\mathcal{S}_1$ collaboratively compute two key metrics while preserving gradient confidentiality.
      First, they calculate each compressed gradient's norm $\Vert \boldsymbol{g}_i^{t*} \Vert$ by applying Algorithm \ref{alg:defence} (SecNorm) to the compressed masked gradients $\boldsymbol{g}_i^{t*} + \boldsymbol{r}_i^{t*}$ and the precomputed encrypted compressed masks $\boldsymbol{c}^{t*}_i$. 
      Second, they determine cosine similarities $\text{cos}_i^t$ between client gradients and the reference gradient using Algorithm \ref{alg:SecCos} (SecCos). 
      Together, these metrics enable the identification of potentially malicious updates without compromising privacy.
\end{enumerate}
Through the JL transformation, we reduce the gradient dimension from $d$ to $k$ (where $k \ll d$), dramatically decreasing computational overhead while preserving the geometric properties essential for Byzantine defense. This compression technique ensures that $\Vert \boldsymbol{g}_i^t \Vert \approx \Vert \boldsymbol{g}_i^{t*} \Vert$, effectively resolving the computation bottleneck inherent in computing norms in the ciphertext domain.

The secure norm and cosine similarity calculations are performed using Algorithms \ref{alg:defence} ($\textbf{SecNorm}(\cdot)$) and \ref{alg:SecCos} ($\textbf{SecCos}(\cdot)$), respectively, which are detailed in the following subsections.

\paragraph{Secure Norm Calculation} Algorithm \ref{alg:defence} presents our protocol for computing the norm of a client gradient $\boldsymbol{g}_i$ in a privacy-preserving manner without exposing the gradient to either server. The protocol enables $\mathcal{S}_1$ to calculate $\Vert \boldsymbol{g}_i \Vert$ using the masked gradient $\boldsymbol{g}_i + \boldsymbol{r}_i$ from $\mathcal{S}_0$ and encrypted mask values $\boldsymbol{c}_i = [\textbf{Enc}(r_{i,0}, pk), \ldots, \textbf{Enc}(r_{i,d-1}, pk)]$ through secure computation. The key insight is to leverage the algebraic identity that decomposes the squared norm of the original gradient:
\begin{equation} \label{eq:norm}
\Vert \boldsymbol{g}_i \Vert^2 = \Vert \boldsymbol{g}_i + \boldsymbol{r}_i \Vert^2 + \Vert \boldsymbol{r}_i \Vert^2 - 2(\boldsymbol{g}_i + \boldsymbol{r}_i) \cdot \boldsymbol{r}_i^\top
\end{equation} 
This design allows us to compute the norm without revealing the actual gradient $\boldsymbol{g}_i$ to either server, thus preserving client privacy while enabling effective Byzantine defense.
For each gradient $\boldsymbol{g}_i$, the secure norm calculation protocol operates as follows:
\begin{enumerate}
  \item $\mathcal{S}_0$ computes encrypted dot product via homomorphic operations:
  \begin{align*}
    c'_{i,j} 
    &= \textbf{ScalarMul}(c_{i,j},\ g_{i,j} + r_{i,j},\ pk) \\
    &= \textbf{Enc}\big(r_{i,j}\cdot (g_{i,j} + r_{i,j}),\ pk\big), 
   \quad \forall j \in [0, d-1]. \\ \\
   c_{\text{sum}} 
    &= \textbf{Add}(\{c'_{i,j}\}_{j < d},\ pk) \\
    &= \textbf{Enc}\bigg(\sum_{j=0}^{d-1} (g_{i,j} + r_{i,j}) \cdot r_{i,j}, pk\bigg).
  \end{align*}

   \item $\mathcal{S}_0$ calculates the squared norm of masked gradient $\Vert \boldsymbol{g}_i + \boldsymbol{r}_i \Vert^2$ and sends both this value and $c_{\text{sum}}$ to $\mathcal{S}_1$.
    
    \item $\mathcal{S}_1$ decrypts to obtain the dot product:
      $$(\boldsymbol{g}_i+ \boldsymbol{r}_i) \cdot (\boldsymbol{r}_i)^\top \leftarrow \textbf{Dec}(c_{\text{sum}}, sk).$$

      \item Finally, $\mathcal{S}_1$ calculates the squared norm $\Vert \boldsymbol{g}_i \Vert^2$ of the client's gradient by applying Equation (\ref{eq:norm}).
  \end{enumerate}

  \begin{algorithm}
    \caption{Privacy-Preserving Gradient Norm Calculation - $\text{SecNorm}(\boldsymbol{g}_i+\boldsymbol{r}_i, \boldsymbol{c}_i)$}\label{alg:defence}
    \DontPrintSemicolon
    \KwIn{$\mathcal{S}_0$ holds masked gradient $\boldsymbol{g}_i + \boldsymbol{r}_i$ and encrypted mask $\boldsymbol{c}^{t}_i = [\textbf{Enc}(r^{t}_{i,j}, pk)]_{j \in[0,d-1]}$; $\mathcal{S}_1$ holds private key $sk$;}

    \KwOut{$\mathcal{S}_1$ obtains gradient norm $\Vert \boldsymbol{g}_i\Vert$ without learning $\boldsymbol{g}_i$}

    \tcc{@$\mathcal{S}_0$: Homomorphic Computation}

    $\mathcal{S}_0$ computes squared norm: $\Vert \boldsymbol{g}_i + \boldsymbol{r}_i \Vert^2 = \sum_{j=0}^{d-1}(g_{i,j}+r_{i,j})^2$\;
    
    \For{$j=0$ \KwTo $d-1$}{
    $c'_{i,j} \leftarrow \textbf{ScalarMul}(c_{i,j}, g_{i,j} + r_{i,j}, pk)$\;
    }
    
    $c_{\text{sum}} \leftarrow \textbf{Add}(\{c'_{i,j}\}_{j < d}, pk)$\;
    \tcp{$= \textbf{Enc}((\boldsymbol{g}_i+ \boldsymbol{r}_i) (\boldsymbol{r}_i)^\top)$}
    $\mathcal{S}_0$ sends $c_{\text{sum}}$ and $\Vert \boldsymbol{g}_i + \boldsymbol{r}_i \Vert^2$ to $\mathcal{S}_1$\;
    
    \tcc{@$\mathcal{S}_1$: Norm Recovery}
    $\mathcal{S}_1$ decrypts: $(\boldsymbol{g}_i+ \boldsymbol{r}_i) \cdot (\boldsymbol{r}_i)^\top \leftarrow \textbf{Dec}(c_{\text{sum}}, sk)$\;
    $\mathcal{S}_1$ computes: $\Vert \boldsymbol{r}_i \Vert^2 = \sum_{j=0}^{d-1}{(r_{i,j})^2}$\;
    
    $\Vert \boldsymbol{g}_i \Vert^2 = \Vert \boldsymbol{g}_i + \boldsymbol{r}_i \Vert^2 + \Vert \boldsymbol{r}_i \Vert^2 - 2 (\boldsymbol{g}_i+ \boldsymbol{r}_i) \cdot (\boldsymbol{r}_i)^\top$\;
    
    \KwRet{$\Vert \boldsymbol{g}_i \Vert = \sqrt{\Vert \boldsymbol{g}_i \Vert^2}$}
  \end{algorithm}

  \paragraph{Secure Cosine Similarity Calculation} Algorithm \ref{alg:SecCos} presents our protocol for computing cosine similarities $\text{cos}_i$ between client gradients $\boldsymbol{g}_i$ and the reference gradient $\boldsymbol{g}_{\text{standard}}$ in a privacy-preserving manner. The protocol leverages the additive secret sharing properties of our masking scheme to enable secure similarity computation without exposing individual client gradients to either server.

  The secure cosine similarity computation operates as follows:
  \begin{enumerate}
    \item $\mathcal{S}_0$ computes the inner product between the masked gradient and the reference gradient: $p_0 = (\boldsymbol{g}_i + \boldsymbol{r}_i) \cdot \boldsymbol{g}_{\text{standard}}^\top$, and sends this value to $\mathcal{S}_1$.
    
    \item $\mathcal{S}_1$ independently computes the inner product between its known mask and the reference gradient: $p_1 = \boldsymbol{r}_i \cdot \boldsymbol{g}_{\text{standard}}^\top$.
    
    \item $\mathcal{S}_1$ recovers the true inner product by exploiting the linearity property: $\boldsymbol{g}_i \cdot \boldsymbol{g}_{\text{standard}}^\top = p_0 - p_1$.
    
    \item Finally, $\mathcal{S}_1$ computes the cosine similarity using the recovered inner product and the gradient norms obtained from Algorithm \ref{alg:defence}:
    \begin{equation}\label{eq:cos}
      \text{cos}_i = \frac{\boldsymbol{g}_i \cdot \boldsymbol{g}_{\text{standard}}^\top}{\Vert \boldsymbol{g}_i \Vert \cdot \Vert \boldsymbol{g}_{\text{standard}} \Vert}
    \end{equation}
  \end{enumerate}

\begin{algorithm}
  \caption{Privacy-Preserving Cosine Similarity - $\text{SecCos}(\boldsymbol{g}_i+\boldsymbol{r}_i, \boldsymbol{r}_i, \boldsymbol{g}_{\text{standard}}, \Vert \boldsymbol{g}_i \Vert)$}\label{alg:SecCos}
    \DontPrintSemicolon
    \KwIn{$\mathcal{S}_0$ holds masked gradient $\boldsymbol{g}_i + \boldsymbol{r}_i$; 
          $\mathcal{S}_1$ holds mask $\boldsymbol{r}_i$, reference gradient $\boldsymbol{g}_{\text{standard}}$, norm $\Vert \boldsymbol{g}_i \Vert$}
    \KwOut{Cosine similarity $\text{cos}_i$ without revealing $\boldsymbol{g}_i$}
    
    \tcc{@$\mathcal{S}_0$: Inner Product Calculation}
    $\mathcal{S}_0$ computes: $p_0 = (\boldsymbol{g}_i + \boldsymbol{r}_i) \cdot \boldsymbol{g}_{\text{standard}}^\top$\;
    $\mathcal{S}_0$ sends $p_0$ to $\mathcal{S}_1$\;
    
    \tcc{@$\mathcal{S}_1$: Similarity Recovery}
    $\mathcal{S}_1$ computes: $p_1 = \boldsymbol{r}_i \cdot \boldsymbol{g}_{\text{standard}}^\top$\;
    $\mathcal{S}_1$ recovers true inner product: $\boldsymbol{g}_i \cdot \boldsymbol{g}_{\text{standard}}^\top = p_0 - p_1$\;
    $\mathcal{S}_1$ calculates $\text{cos}_i$ according to Equation (\ref{eq:cos})\;
    
    \KwRet{$\text{cos}_i$}
\end{algorithm}
\subsection{Secure Aggregation}
Algorithm \ref{alg:SecAgg} details the secure aggregation process, which combines local gradients using a Byzantine-robust weighted averaging scheme. The aggregation weights are derived from the cosine similarities and gradient norms computed during the defense phase, ensuring that malicious updates receive minimal influence while preserving the contributions of legitimate clients.
The secure aggregation process operates in three phases. First, $\mathcal{S}_1$ computes trust scores using $\text{TS}_i = \max(0, \text{cos}_i)$ and calculates the aggregation weights for each gradient $\boldsymbol{g}_i$:
$$\omega_i = \frac{\text{TS}_i}{\sum_{j=1}^{n} \text{TS}_j} \cdot \frac{\Vert\boldsymbol{g}_{\text{standard}}\Vert}{\Vert \boldsymbol{g}_i \Vert}$$

Second, $\mathcal{S}_1$ computes the weighted mask sum $\boldsymbol{m} = \sum_{i=1}^{n} \omega_i \boldsymbol{r}_i$ and sends both the aggregation weights $\{\omega_i\}_{i \in [n]}$ and the weighted mask sum $\boldsymbol{m}$ to $\mathcal{S}_0$. 

Finally, $\mathcal{S}_0$ performs the secure weighted aggregation:
$$\boldsymbol{g}_{\text{global}} = \sum_{i=1}^{n} \omega_i (\boldsymbol{g}_i + \boldsymbol{r}_i) - \boldsymbol{m}$$
This recovers the true weighted average of client gradients $\sum_{i=1}^{n} \omega_i \boldsymbol{g}_i$ without either server learning individual gradient values. The final aggregated gradient $\boldsymbol{g}_{\text{global}}$ is then broadcast to all clients for the model update phase.

\begin{algorithm}
  \caption{Secure Robust Aggregation - $\text{SecAgg}(\{\boldsymbol{g}_i+\boldsymbol{r}_i\}, \{\boldsymbol{r}_i\}, \{\text{cos}_i\}, \{\Vert \boldsymbol{g}_i \Vert\})$}\label{alg:SecAgg}
    \DontPrintSemicolon
    \KwIn{$\mathcal{S}_0$ holds masked gradients $\{\boldsymbol{g}_i + \boldsymbol{r}_i\}_{i<n}$; 
          $\mathcal{S}_1$ holds masks $\{\boldsymbol{r}_i\}_{i<n}$, reference gradient $\boldsymbol{g}_{\text{standard}}$, 
          similarities $\{\text{cos}_i\}_{i<n}$, norms $\{\Vert \boldsymbol{g}_i \Vert\}_{i<n}$}
    \KwOut{Global gradient $\boldsymbol{g}_{\text{global}}$}
    
    \tcc{@$\mathcal{S}_1$: Weight Computation}
    $\mathcal{S}_1$ computes for each client $i$:\;
    $\text{TS}_i = \max(0, \text{cos}_i)$\;
    $\omega_i = \frac{\text{TS}_i}{\sum_{j=1}^{n} \text{TS}_j} \cdot \frac{\Vert\boldsymbol{g}_{\text{standard}}\Vert}{\Vert \boldsymbol{g}_i \Vert}$ \;
    
    \tcc{@$\mathcal{S}_0$: Secure Weighted Aggregation}
    $\mathcal{S}_1$ computes weighted mask sum: $\boldsymbol{m} = \sum_{i=1}^{n} \omega_i \boldsymbol{r}_i$\;
    $\mathcal{S}_1$ sends weights $\{\omega_i\}_{i \in [n]}$ and $\boldsymbol{m}$ to $\mathcal{S}_0$\;
    $\mathcal{S}_0$ aggregates: $\boldsymbol{g}_{\text{global}} = \sum_{i=1}^{n} \omega_i (\boldsymbol{g}_i + \boldsymbol{r}_i) - \boldsymbol{m}$\;
    
    \KwRet{$\boldsymbol{g}_{\text{global}}$}
\end{algorithm}

\subsection{Update Phase}
Finally, $\mathcal{S}_0$ broadcasts the aggregated global gradient $\boldsymbol{g}_{\text{global}}^t$ to all clients. Each client updates its local model using the standard gradient descent update rule: $\boldsymbol{W}_i^{t+1} = \boldsymbol{W}^{t} - \eta \cdot \boldsymbol{g}_{\text{global}}^t$, where $\eta$ is the learning rate (lines 15-19 in Algorithm \ref{alg:workflow}).

\section{Theoretical Analysis} \label{sec:theoretical_analysis}

\subsection{Security Analysis}
In this subsection, we employ the standard real-world/ideal-world simulation paradigm \cite{howtosimulate} from secure multiparty computation theory to formally establish that adversaries gain no additional information about honest clients' gradients beyond what is explicitly revealed through legitimate protocol outputs.

Following this paradigm, we model a security environment $\mathcal{E}$ where both real and ideal world executions receive identical inputs. Our protocol $\Pi$ is defined as a probabilistic polynomial-time algorithm that processes the set of client gradients $\{\boldsymbol{g}_i\}_{i \in [n]}$ to compute an aggregated result. The real-world execution $\text{Real}_{\Pi}(\kappa_1, \mathcal{A}, \{\boldsymbol{g}_i\}_{i \in [n]})$ denotes the output distribution when adversary $\mathcal{A}$ interacts with honest parties following $\Pi$, with all parties receiving their respective inputs and security parameter $\kappa_1$. In contrast, the ideal-world execution $\text{Ideal}_{\mathcal{F},\mathtt{Sim}}(\kappa_1, \mathcal{A}, \{\boldsymbol{g}_i\}_{i \in [n]})$ represents the output distribution when a simulator $\mathtt{Sim}$ interacts with the trusted functionality $\mathcal{F}$, which performs the intended computation with perfect security guarantees by securely processing all inputs without revealing any additional information beyond the legitimate outputs.

\begin{definition}
  \label{def:real_ideal}
  A protocol $\Pi$ securely realizes functionality $\mathcal{F}$ in the presence of semi-honest adversaries if for every probabilistic polynomial-time adversary $\mathcal{A}$ in the real world, there exists a probabilistic polynomial-time simulator $\mathtt{Sim}$ in the ideal world such that for any input set $\{\boldsymbol{g}_i\}_{i \in [n]}$:
  \begin{equation}
    \text{Real}_{\Pi}(\kappa_1, \mathcal{A}, \{\boldsymbol{g}_i\}_{i \in [n]}) \stackrel{c}{\approx} \text{Ideal}_{\mathcal{F},\mathtt{Sim}}(\kappa_1, \mathcal{A}, \{\boldsymbol{g}_i\}_{i \in [n]})
  \end{equation}
  where $\stackrel{c}{\approx}$ denotes computation indistinguishability as a function of the security parameter $\kappa_1$.
\end{definition}

\begin{lemma} \label{lem:paillier}
    \textnormal{(From~\cite{Paillier})} Paillier encryption is semantically secure under the honest-but-curious setting.
\end{lemma}

To establish the security of our protocol, we consider an adversary $\mathcal{A}$ that corrupts either one server (but not both) or a subset of clients. We analyze the adversary's view in both real and ideal worlds by employing the hybrid argument technique, where we construct a sequence of hybrid worlds that gradually transition from the real-world execution to the ideal-world execution. Each hybrid world corresponds to a different stage of the protocol, and we demonstrate that the views of the adversary in each hybrid world are computationally indistinguishable. By the transitivity of computational indistinguishability, this ensures that the adversary's view in the real world is computationally indistinguishable from the ideal world, thereby guaranteeing that no information about honest clients' gradients can be extracted beyond what is legitimately revealed through the protocol's outputs.

\begin{theorem}
  \label{thm:indistinguishability}
  (Gradient Confidentiality Against Semi-honest Servers):
  In the presence of a semi-honest adversary $\mathcal{A}$ that corrupts a single server $\mathcal{S}_x$ (where $x \in \{0,1\}$), our protocol $\Pi$ securely realizes the ideal functionality $\mathcal{F}$ for Byzantine-robust federated aggregation, as defined in Definition \ref{def:real_ideal}. 

\end{theorem}

\begin{proof}
  We construct a simulator $\mathtt{Sim}$ that can generate a view computationally indistinguishable from the real-world execution without access to honest clients' private gradients. We analyze both potential corruption cases separately:

  \textbf{Case 1: $\mathcal{A}$ corrupts $\mathcal{S}_0$ (with $\mathcal{S}_1$ honest)}

  In this case, the adversary $\mathcal{A}$ controlling $\mathcal{S}_0$ observes:
  \begin{itemize}
    \item Masked gradients $\{\boldsymbol{g}_i + \boldsymbol{r}_i\}_{i \in [n]}$ received from clients.
    \item Encrypted mask projections from $\mathcal{S}_1$:
    $$\{\boldsymbol{c}^*_i\}_{i \in [n]} = \{\textbf{Enc}(r^*_{i,j}, pk)\}_{i \in [n], j \in [0,k-1]}.$$
    \item Aggregation weights $\{\omega_i\}_{i \in [n]}$ and weighted mask sum $\boldsymbol{m}= \sum_{i=1}^{n} \omega_i \boldsymbol{r}_i$ from $\mathcal{S}_1$.
  \end{itemize}

  We construct the following sequence of hybrid worlds:

  \textbf{Hybrid 0 ($H_0$)}: The real protocol execution where $\mathcal{S}_0$ receives actual protocol messages.

  \textbf{Hybrid 1 ($H_1$)}: Identical to $H_0$, except that the simulator replaces the masked gradients $\{\boldsymbol{g}_i + \boldsymbol{r}_i\}_{i \in [n]}$ with uniformly random vectors $\{\boldsymbol{\phi}_i\}_{i \in [n]}$.

  \textit{Indistinguishability of $H_0$ and $H_1$}: Since each mask $\boldsymbol{r}_i$ is generated from a random seed $s_i$ using a PRNG, it follows a uniform distribution and is independent of $\boldsymbol{g}_i$. The masked gradient $\boldsymbol{g}_i + \boldsymbol{r}_i$ achieves perfect secrecy through the one-time pad property, making it uniformly distributed and statistically independent of the original gradient $\boldsymbol{g}_i$. Therefore, the distributions of $\{\boldsymbol{g}_i + \boldsymbol{r}_i\}_{i \in [n]}$ and $\{\boldsymbol{\phi}_i\}_{i \in [n]}$ are computationally indistinguishable, establishing perfect indistinguishability between $H_0$ and $H_1$.

  \textbf{Hybrid 2 ($H_2$)}: Identical to $H_1$, except that the simulator generates encrypted mask projections $\{\boldsymbol{c}^*_i\}_{i \in [n]}$ as encryptions of random values $\{\boldsymbol{\chi}_i\}_{i \in [n]}$ using the public key $pk$.

  \textit{Indistinguishability of $H_1$ and $H_2$}: By the semantic security of the Paillier cryptosystem from Lemma \ref{lem:paillier}, the encrypted values $\{\boldsymbol{c}^*_i = \text{Enc}(\boldsymbol{r}^*_i, pk)\}_{i \in [n]}$ and $\{\text{Enc}(\boldsymbol{\chi}_i, pk)\}_{i \in [n]}$ are computationally indistinguishable to any adversary without the secret key $sk$. Therefore, $H_1$ and $H_2$ are computationally indistinguishable.

  \textbf{Hybrid 3 ($H_3$)}: Identical to $H_2$, except that the simulator generates random aggregation weights $\{\omega_i\}_{i \in [n]}$ subject to $\sum_{i=1}^{n} \omega_i = 1$ and a random vector $\boldsymbol{\mu}$ as the weighted mask sum.

  \textit{Indistinguishability of $H_2$ and $H_3$}: The weights $\{\omega_i\}_{i \in [n]}$ in the real protocol are derived from cosine similarities and norm values that are computed securely without $\mathcal{S}_0$ learning their actual values. From $\mathcal{S}_0$'s perspective, these weights appear random subject to normalization constraints. Similarly, since each mask $\boldsymbol{r}_i$ is uniformly random, the weighted sum $\boldsymbol{m} = \sum_{i=1}^{n} \omega_i \boldsymbol{r}_i$ is computationally indistinguishable from a random vector $\boldsymbol{\mu}$ of the same dimension. Thus, $H_2$ and $H_3$ are computationally indistinguishable.

  $H_3$ corresponds to the simulator $\mathtt{Sim}$'s output in the ideal world. By the transitivity of computation indistinguishability, the real-world view ($H_0$) is indistinguishable from the ideal-world view ($H_3$).

  \textbf{Case 2: $\mathcal{A}$ corrupts $\mathcal{S}_1$ (with $\mathcal{S}_0$ honest)}

  In this case, the adversary $\mathcal{A}$ controlling $\mathcal{S}_1$ observes:
  \begin{itemize}
    \item Mask seeds $\{s_i\}_{i \in [n]}$ received from clients.
    \item Inner products $\{p_0^i = (\boldsymbol{g}_i + \boldsymbol{r}_i) \cdot \boldsymbol{g}_{\text{standard}}^\top\}_{i \in [n]}$ from $\mathcal{S}_0$.
    \item Squared norms $\{d_i^2 = \Vert \boldsymbol{g}_i^*+\boldsymbol{r}_i^*\Vert^2\}_{i \in [n]}$ from $\mathcal{S}_0$.
    \item Encrypted dot products $\{c_{\text{sum}}^i = \textbf{Enc}((\boldsymbol{g}_i^*+\boldsymbol{r}_i^*)\cdot\boldsymbol{r}_i^*)\}_{i \in [n]}$ from $\mathcal{S}_0$.
  \end{itemize}

  We construct the following hybrids:

  \textbf{Hybrid 0 ($H_0$)}: The real protocol execution.

  \textbf{Hybrid 1 ($H_1$)}: Identical to $H_0$, except that the simulator uses randomly generated seeds $\{\mu_i\}_{i \in [n]}$ instead of the actual masks $\{s_i\}_{i \in [n]}$.

  \textit{Indistinguishability of $H_0$ and $H_1$}: Since the masks $\{s_i\}_{i \in [n]}$ are uniformly random by design, replacing them with another set of uniformly random vectors $\{\mu_i\}_{i \in [n]}$ produces an identical distribution, making $H_0$ and $H_1$ perfectly indistinguishable.

  \textbf{Hybrid 2 ($H_2$)}: Identical to $H_1$, except that the simulator generates synthetic values using carefully constructed vectors $\{\boldsymbol{\psi}_i\}_{i \in [n]}$ instead of the actual client gradients $\{\boldsymbol{g}_i\}_{i \in [n]}$.

  The simulator constructs vectors $\{\boldsymbol{\psi}_i\}_{i \in [n]}$ that preserve the exact properties observable by $\mathcal{S}_1$: the norm $\Vert \boldsymbol{g}_i^* \Vert$ and inner product $\boldsymbol{g}_i \cdot \boldsymbol{g}_{\text{standard}}^\top$. For each gradient, the simulator:
  \begin{enumerate}
    \item Generates random vector $\boldsymbol{u} \in \mathbb{R}^d$
    \item Projects $\boldsymbol{u}$ orthogonal to $\boldsymbol{g}_{\text{standard}}$: 
    $$\boldsymbol{u}_{\perp} = \boldsymbol{u} - \frac{\boldsymbol{u} \cdot \boldsymbol{g}_{\text{standard}}^\top}{\|\boldsymbol{g}_{\text{standard}}\|^2}\boldsymbol{g}_{\text{standard}}$$
    \item Normalizes: $\hat{\boldsymbol{u}}_{\perp} = \frac{\boldsymbol{u}_{\perp}}{\|\boldsymbol{u}_{\perp}\|}$
    \item Constructs $\boldsymbol{\psi}_i$ with matching norm and inner product: 
    $$\boldsymbol{\psi}_i = \frac{\rho \cdot \boldsymbol{g}_{\text{standard}}}{\|\boldsymbol{g}_{\text{standard}}\|^2} + \sqrt{d^2 - \frac{\rho^2}{\|\boldsymbol{g}_{\text{standard}}\|^2}} \cdot \hat{\boldsymbol{u}}_{\perp}$$ 
    where $\rho = \boldsymbol{g}_i \cdot \boldsymbol{g}_{\text{standard}}^\top$
  \end{enumerate}

  Using vectors $\{\boldsymbol{\psi}_i\}_{i \in [n]}$ and random masks $\{\boldsymbol{\phi}_i\}_{i \in [n]} = \{G(\mu_i)\}_{i \in [n]}$, the simulator generates synthetic values:
    \begin{align*}
    \left\{\hat{p}_0^i\right\}_{i \in [n]} &= \left\{(\boldsymbol{\psi}_i + \boldsymbol{\phi}_i) \cdot \boldsymbol{g}_{\text{standard}}^\top\right\}_{i \in [n]} \\
    \left\{\hat{d}_i^2\right\}_{i \in [n]} &= \left\{\Vert \boldsymbol{\psi}_i^*+\boldsymbol{\phi}_i^*\Vert^2\right\}_{i \in [n]} \\
    \left\{\hat{c}_{\text{sum}}^i\right\}_{i \in [n]} &= \left\{\textbf{Enc}((\boldsymbol{\psi}_i^*+\boldsymbol{\phi}_i^*)\cdot(\boldsymbol{\phi}_i^*)^\top)\right\}_{i \in [n]}
    \end{align*}

  \textit{Indistinguishability of $H_1$ and $H_2$}: By our construction, the synthetic vectors $\boldsymbol{\psi}_i$ preserve exactly the properties of $\boldsymbol{g}_i$ that are observable by $\mathcal{S}_1$ through the protocol: specifically, $\boldsymbol{\psi}_i \cdot \boldsymbol{g}_{\text{standard}}^\top = \boldsymbol{g}_i \cdot \boldsymbol{g}_{\text{standard}}^\top$ and $\Vert \boldsymbol{\psi}_i^* \Vert = \Vert \boldsymbol{g}_i^* \Vert$.
  Consequently, after $\mathcal{S}_1$ processes the values received in both hybrids (subtracting the mask from $p_0^i$ and decrypting $c_{\text{sum}}^i$), the resulting values are mathematically identical. For any $\boldsymbol{\psi}_i$ constructed using our method, we have $\boldsymbol{\psi}_i \cdot \boldsymbol{g}_{\text{standard}}^\top = \boldsymbol{g}_i \cdot \boldsymbol{g}_{\text{standard}}^\top$ and $\Vert \boldsymbol{\psi}_i^* \Vert = \Vert \boldsymbol{g}_i^* \Vert$ by design.

  Since $\mathcal{S}_1$ only observes these specific geometric properties (inner products and norms) through the protocol, and there exists an infinite-dimensional subspace of vectors maintaining these same properties (due to our construction with randomly chosen orthogonal components), no additional information about the original gradients $\boldsymbol{g}_i$ is revealed. The views in $H_1$ and $H_2$ are thus perfectly indistinguishable to $\mathcal{S}_1$.
  Therefore, the view of $\mathcal{S}_1$ in $H_1$ is perfectly indistinguishable from its view in $H_2$.

  Since $H_2$ corresponds to the simulator's output in the ideal world, and $H_0$ and $H_2$ are computationally indistinguishable, the theorem holds.
\end{proof}

    Through our hybrid arguments, we've demonstrated that the adversary's view in the real protocol is computationally indistinguishable from its view in the ideal world. This proves that adversaries learn nothing about honest clients' gradients beyond the legitimate outputs.

\subsection{Complexity Analysis} \label{sec:complexity_analysis}
In this subsection, we analyze the computation and communication complexity of our proposed protocol, using ShieldFL \cite{ShieldFL} as the baseline due to its similar privacy guarantees. Communication complexity is measured by the number of bits transferred between clients and servers, while computation complexity is quantified by the number of fundamental operations required. We denote $T_{\text{exp}}$ as the time for a single modular exponentiation over a finite field of length $2\kappa_1$, $T_{\text{mul}}$ as the time for modular multiplication over the same field, and $T_{\text{add}}$ as the time for modular addition operations over a smaller field of length $\kappa_2$. 
Note that the computational overhead of compression operations in our scheme is negligible, as they involve only real-domain matrix multiplications that are computationally insignificant compared to cryptographic operations over finite fields.
For instance, according to our experimental measurements, a single modular exponentiation operation requires approximately 2.9ms over a 1024-bit field, while a real-domain multiplication requires only $5 \times 10^{-5}$ms. Additionally, the compression operations can be further accelerated by modern GPU hardware.

\subsubsection{Computation Complexity}
Table \ref{tab:computational_overhead} presents a comprehensive comparison of computational complexity between our proposed approach and ShieldFL. 

On the client-side, our scheme demonstrates significantly reduced computational overhead by implementing simple masking operations instead of requiring clients to perform expensive homomorphic encryption operations as in ShieldFL. This makes our scheme particularly suitable for resource-constrained client devices. The most substantial efficiency gains appear on the server-side. While ShieldFL requires $(18dn + 26n + 2d)T_{\text{exp}} + (9dn + 9n + 2d)T_{\text{mul}}$ operations, our scheme reduces computational overhead in three key ways. First, our basic scheme decreases complexity to $(dn + n)T_{\text{exp}} + (dn + n)T_{\text{mul}}$ by strategically combining additive masking with minimal homomorphic operations. Second, our compressed scheme further reduces computation requirements to $(kn + n)T_{\text{exp}} + (kn + n)T_{\text{mul}}$, where $k \ll d$, representing an asymptotic improvement from $O(dn)$ to $O(kn)$. Finally, our scheme enables preprocessing of cryptographic operations, as masks can be generated and encrypted offline, further reducing real-time computation demands.

\begin{table*}[!t]
  \caption{computational overhead of Schemes\label{tab:computational_overhead}}
  \centering
  \begin{tabular}{|c||c|c|c|}
    \hline
    \textbf{Scheme} & \textbf{Clients} & \textbf{Server (Online)} & \textbf{Server (Offline)} \\
    \hline
    ShieldFL & $2dnT_{\text{exp}} + dnT_{\text{mul}}$ & $(18dn + 26n + 2d)T_{\text{exp}} + (9dn + 9n + 2d)T_{\text{mul}}$ & -  \\
    \hline
    Ours w/o compression & $dnT_{\text{add}}$ & $(dn + n)T_{\text{exp}} + (dn + n)T_{\text{mul}}$ &  $2dnT_{\text{exp}} + dnT_{\text{mul}}$ \\
    \hline
    Ours with compression & $dnT_{\text{add}}$ & $(kn + n)T_{\text{exp}} + (kn + n)T_{\text{mul}}$ & $2knT_{\text{exp}} + knT_{\text{mul}}$ \\

    \hline
  \end{tabular}
  \end{table*}

\subsubsection{Communication Complexity}
Table \ref{tab:communication_overhead} presents a comprehensive comparison of communication overhead between our protocol and ShieldFL. Our scheme achieves substantial efficiency gains while maintaining equivalent security guarantees. On the client side, our approach requires only $dn\kappa_2$ bits of communication, which is significantly more efficient than ShieldFL's $2dn\kappa_1$ bits per client, since $\kappa_2 < \kappa_1$ in our design.
The most substantial improvements emerge in server-to-server communication. While ShieldFL requires $(12dn + 12n + 4d)\kappa_1$ bits of inter-server communication, our uncompressed scheme reduces this to $(2dn + 4n + d)\kappa_1$ bits by strategically combining additive masking with minimal homomorphic operations, thereby reducing the need for expensive ciphertext transmissions. 
When compression is applied, our protocol's communication requirements further decrease to $(2kn + 4n + d)\kappa_1$ bits, where $k \ll d$. This represents an improvement from $O(dn)$ to $O(kn)$ in communication complexity.

\begin{table}[!t]
  \caption{Communication overhead of Schemes\label{tab:communication_overhead}}
  \centering
  \begin{tabular}{|c||c|c|}
    \hline
    \textbf{Scheme} & \textbf{Clients} & \textbf{Server} \\
    \hline
    ShieldFL & $2dn\kappa_1$ & $(12dn + 12n + 4d)\kappa_1$ \\
    \hline
    Ours without compression& $dn\kappa_2$ & $(2dn + 4n + d)\kappa_1$ \\
    \hline
    Ours with compression & $dn\kappa_2$ & $(2kn + 4n + d)\kappa_1$ \\
    \hline
  \end{tabular}
  \end{table}

\subsection{The influence of the JL Transformation} \label{sec:analysis_jl}
As demonstrated in our complexity analysis (Section \ref{sec:complexity_analysis}), our protocol achieves substantial computational and communication efficiency improvements through JL transformation-based gradient compression. The JL transformation projects high-dimensional gradients from $\mathbb{R}^d$ to a lower-dimensional space $\mathbb{R}^k$ where $k \ll d$, while maintaining the geometric relationships necessary for Byzantine detection.

According to Lemma \ref{lem:JL}, the compressed gradient $\boldsymbol{g}_i^* \in \mathbb{R}^k$ preserves key properties of the original gradient $\boldsymbol{g}_i$ within provable error bounds with probability at least $1-\delta$:
\begin{equation*}
  (1-\epsilon)\Vert g_i \Vert^2 \leq \Vert g_i^* \Vert^2 \leq (1+\epsilon)\Vert g_i \Vert^2
\end{equation*} 
Similarly, cosine similarity between gradients is preserved with bounded error:
\begin{equation*}
  \frac{1}{(1-\epsilon)}\frac{\boldsymbol{g}_i \cdot \boldsymbol{g}_{\text{standard}}^\top}{\Vert \boldsymbol{g}_i \Vert \cdot \Vert \boldsymbol{g}_{\text{standard}} \Vert} \leq \text{cos}_i \leq \frac{1}{(1+\epsilon)}\frac{\boldsymbol{g}_i \cdot \boldsymbol{g}_{\text{standard}}^\top}{\Vert \boldsymbol{g}_i \Vert \cdot \Vert \boldsymbol{g}_{\text{standard}} \Vert}
\end{equation*}

Figure \ref{fig:cosine_error} illustrates the maximum cosine similarity error between compressed and original gradients for dimensions ranging from $100,000$ to $10,000,000$, confirming that errors remain within theoretical bounds even at extreme compression ratios. Notably, larger model architectures experience smaller errors at the same compression ratio, making our approach increasingly effective as model size grows.

To optimize the trade-off between computational efficiency and Byzantine robustness, we select compression ratios $k/d$ that maintain a theoretical error bound of $\epsilon \leq 0.2$ with confidence parameter $\delta = 0.01$. We adapt the compression ratio to each model's complexity: $k/d=0.01$ for ResNet20 (0.27M parameters), $k/d=0.001$ for MobileNetV1 (3.2M parameters), and $k/d=0.0001$ for ResNet18 (11.2M parameters). As demonstrated in Figure \ref{fig:cr_acc}, this configuration maintains effective Byzantine defense while enabling substantial computational savings.

\begin{figure}[!t]
  \centering
  \includegraphics[width=0.7\linewidth]{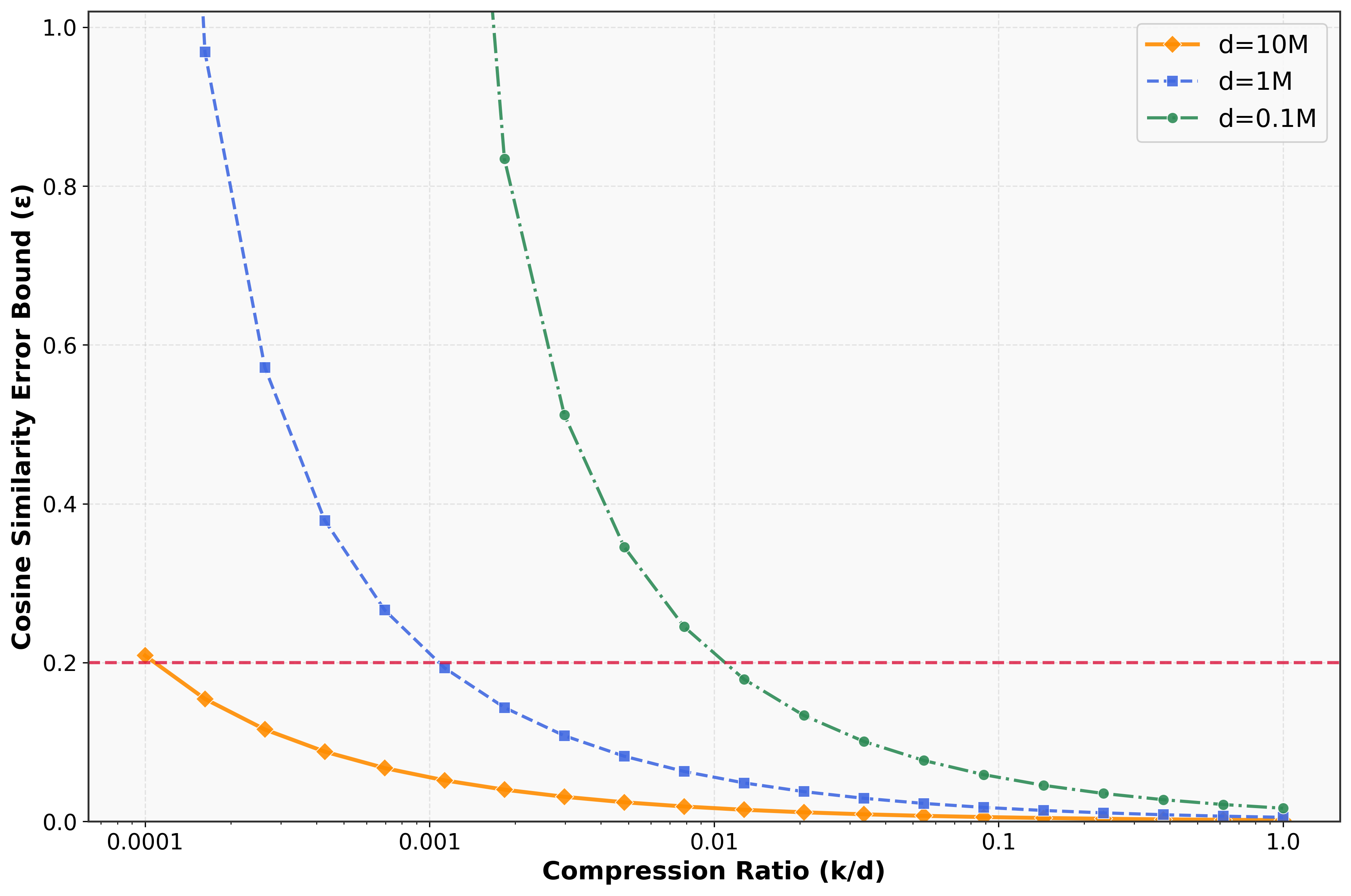}
  \caption{The cosine error between the compressed gradient and the original gradient.}
  \label{fig:cosine_error}
\end{figure}

\section{Experiments} \label{sec:experiments}

\subsection{Experimental Setup}
\subsubsection{Enviromental Settings}
All experiments are carried out in consistent settings to ensure fair comparisons. The hardware environment consists of an AMD Ryzen Threadripper 3970X CPU with 128 GB RAM, and model training was accelerated using an NVIDIA GeForce RTX 3080 Ti GPU. We implement all approaches using Python 3.10.13 and PyTorch 2.1.0 \cite{paszke2019pytorch}. 

\subsubsection{Datasets, Models and Hyperparameters}
Following prior work \cite{RFed, SplitAgg,LSFL}, we evaluate our approach on three standard image classification datasets: SVHN \cite{svhn}, CIFAR-10, and CIFAR-100 \cite{cifar}. CIFAR-10 contains 60,000 32×32 color images across 10 classes. SVHN (Street View House Numbers) consists of 600,000 32×32 color images of digits (0-9). The more challenging CIFAR-100 dataset contains 100 classes. The model architectures we used range from ResNet20 (0.27 million parameters) to MobileNetV1 \cite{mobilenet} (3.2 million parameters) and ResNet18 (11.2 million parameters), allowing us to evaluate our compression approach across different parameter scales.

We conduct the experiments with $100$ clients, where $10\%$ are selected for each round's training. All models are trained using SGD optimizer. The hyperparameters including training parameters and security parameters for each dataset are summarized in Table \ref{tab:hyperparameters}.

\begin{table}[!t]
  \caption{Hyperparameters\label{tab:hyperparameters}}
  \centering
  \begin{tabular}{|c|c|c|c|}
    \hline
    \textbf{Datasets} & SVHN & Cifar10 & Cifar100 \\
    \hline
    Model & MobileNetV1 & ResNet20 & ResNet18 \\
    \hline
    Learning Rate & $0.01$ & $0.1$ & $0.1$ \\
    \hline
    Batch Size& $32$ & $32$ & $10$\\
    \hline
    Compression Ratio & $0.001$ & $0.01$ & $0.0001$ \\
    \hline
    Security Parameter $\kappa_1$ & $512$ & $512$ & $512$ \\
    \hline
    Bit Length for Mask $\kappa_2$ & $64$ & $64$ & $64$ \\
    \hline
  \end{tabular}
\end{table}

\subsubsection{Poisoning Attacks}
We evaluated our scheme against various poisoning attacks following prior work \cite{LSFL,Li2024,SplitAgg,ShieldFL}:
\begin{itemize}
  \item \textbf{Sign-Flipping Attack:} Adversaries submit $-\boldsymbol{g_i}$ instead of true gradients $\boldsymbol{g_i}$.

  \item \textbf{Label-Flipping Attack:} Adversaries train on corrupted labels (e.g., mapping label $x$ to $9-x$ for CIFAR-10/SVHN, $x$ to $99-x$ for CIFAR-100).

  \item \textbf{Gaussian Noise Attack:} Adversaries add noise sampled from $\mathcal{N}(0, 4\sigma^2 I)$ to gradients.
  
  \item \textbf{Scaling Attack:} Adversaries multiply gradients by large factors ($c = 6$).

  \item \textbf{Optimization-Based Attacks:} Advanced attacks targeting specific defenses like Multi-Krum and Trimmed-Mean~\cite{local_poisoning}, denoted as Fang-MKrum and Fang-Trimmed-Mean respectively.
  
  \item \textbf{Min-Max/Min-Sum Attacks:} Strategic attacks minimizing distance metrics to evade Byzantine-robust aggregation \cite{manipulating}.
\end{itemize}

\subsubsection{Baselines}
We compare our approach against several baseline methods across different categories to provide a comprehensive evaluation. For robust FL methods, we include representative Byzantine-robust aggregation techniques: Krum \cite{krum}, Trimmed Mean \cite{trimmedmean}, and FLTrust \cite{fltrust}. For robust PPFL approaches, we evaluate against state-of-the-art methods including ShieldFL \cite{ShieldFL}, which employs dual-trapdoor homomorphic encryption with cosine similarity that achieves strong privacy guarantees, and two recent efficient Byzantine-robust PPFL methods: Split Aggregation \cite{SplitAgg} and LSFL \cite{LSFL}, which utilize SMC with RSVD and Euclidean distance metrics, respectively, representing the most efficient approaches currently available. Additionally, we include the standard FedAvg \cite{FedAvg} as a non-robust baseline to demonstrate the necessity of Byzantine defense mechanisms.

\subsection{Experimental Results}
We evaluate our scheme against baseline approaches across three dimensions: model accuracy under Byzantine attacks, computational efficiency (per-round execution time), and communication overhead (data transmission volume). We also analyze the impact of compression ratios (1.0 to 0.0001) on accuracy and efficiency across different model architectures.
\subsubsection{Byzantine Defence Performance}
Figure \ref{fig:byzantine_defence} illustrates our scheme's robustness compared to baseline approaches under various Byzantine attacks with adversarial clients ranging from 0\% to 40\% on CIFAR-10 and SVHN datasets. In benign environments (0\% Byzantine clients), most schemes maintain accuracy comparable to FedAvg, with only Split Aggregation and Krum showing slight degradation. However, as the proportion of adversarial clients increases, most baseline approaches experience substantial accuracy drops under at least one attack scenario. In contrast, only FLTrust and our scheme consistently maintain stable performance across all attack scenarios, even with 40\% adversarial clients. Notably, our approach achieves this robustness while implementing efficient dimension compression, demonstrating that our compression technique preserves the security guarantees of FLTrust without sacrificing defense capabilities.

\begin{figure}[!t]
  \centering
  \subfloat[CIFAR-10 with ResNet20]
  {\includegraphics[width=0.45\textwidth]{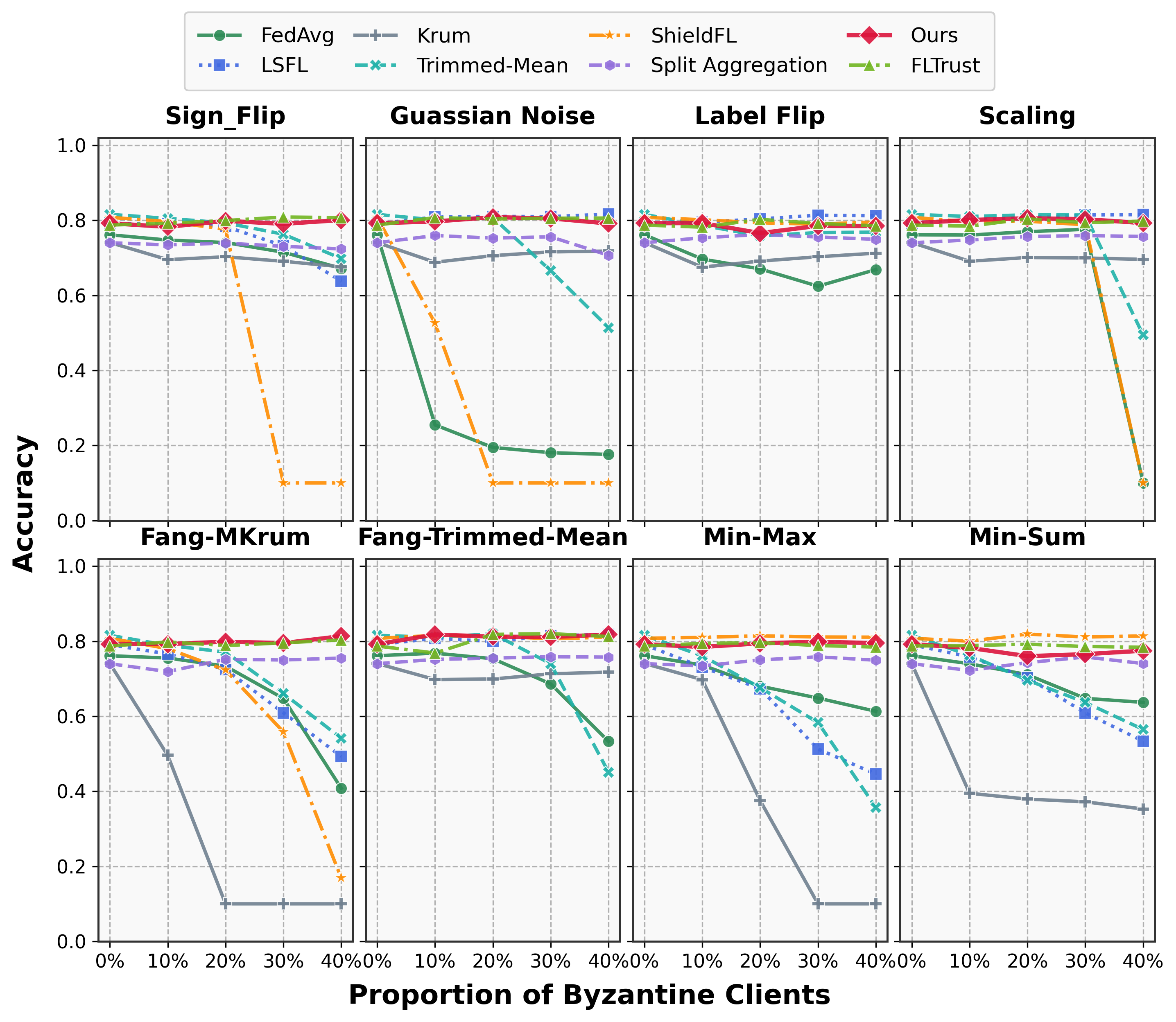}%
  \label{fig:byzantine_defence_resnet}}
  \vspace{0mm}
  \subfloat[SVHN with MobileNetV1]
  {\includegraphics[width=0.45\textwidth]{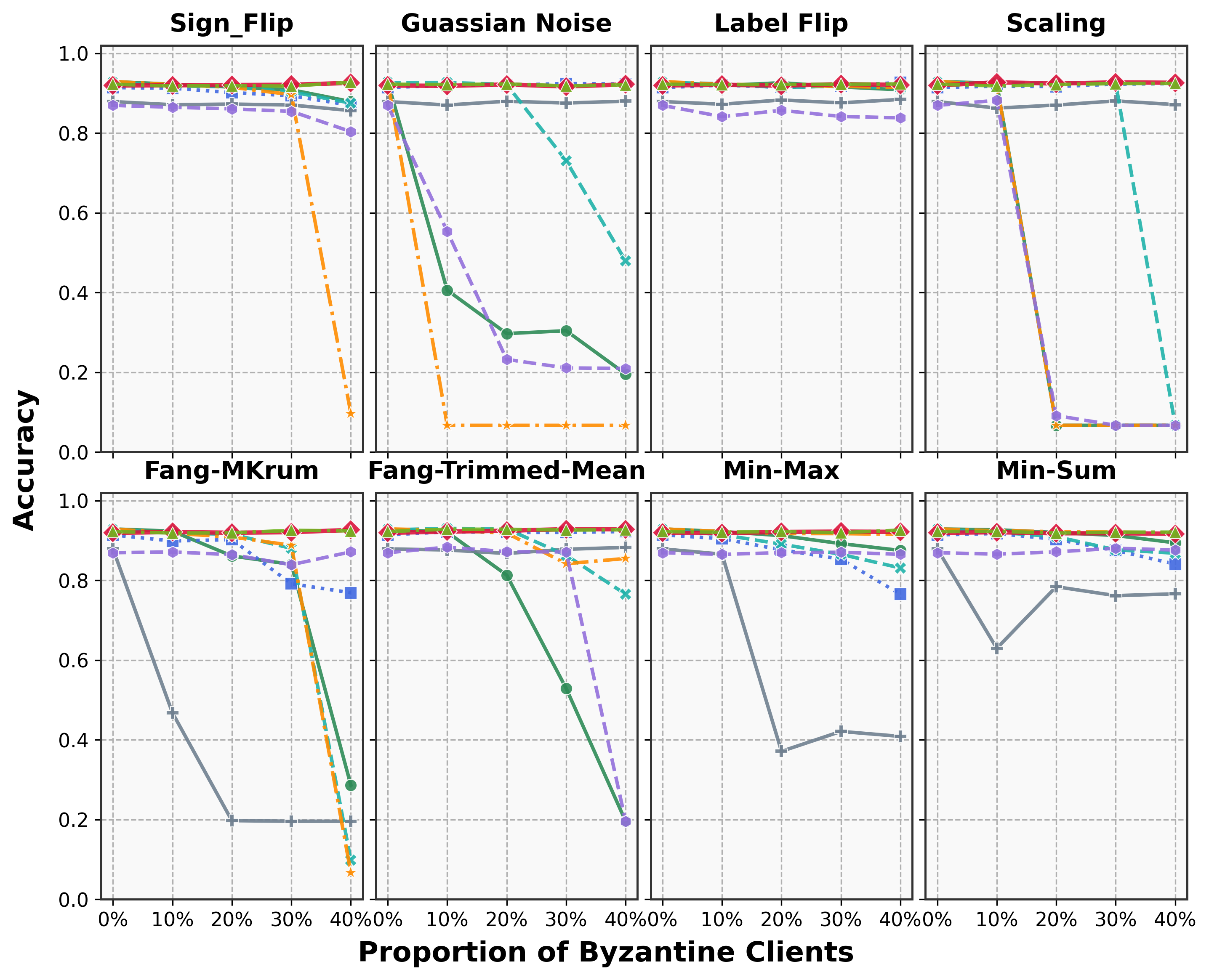}%
  \label{fig:byzantine_defence_mobilenet}}
  \vspace{0mm}
  \subfloat[CIFAR-100 with ResNet18]
  {\includegraphics[width=0.45\textwidth]{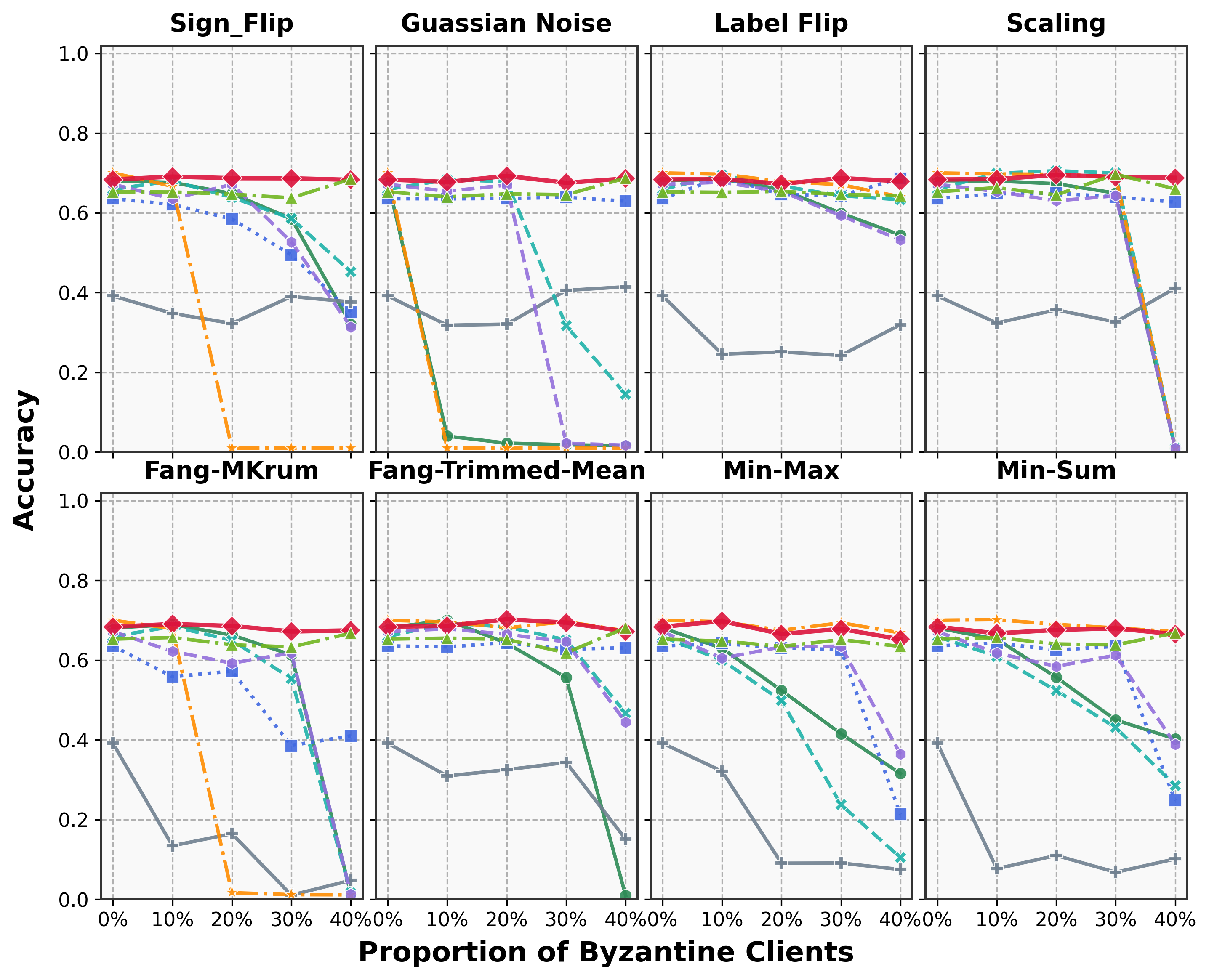}%
  \label{fig:byzantine_defence_resnet18}}
  \vspace{-1mm}
  \caption{Byzantine defence performance under different attack scenarios with varying numbers of Byzantine clients. We evaluate the robustness of our scheme by comparing the model accuracy across different attack scenarios with increasing percentages of adversarial clients (from 0\% to 40\% of the total).}
  \label{fig:byzantine_defence}
\end{figure}

\subsubsection{Comparison of efficiency}
    We compared our scheme's computational and communication efficiency against baseline approaches. Based on model size, we used compression ratios of 0.01 for ResNet20, 0.001 for MobileNetV1, and 0.0001 for ResNet18 (justified in Section~\ref{sec:analysis_jl}), and included our non-compression version as a baseline.

      Figure \ref{fig:efficiency} demonstrates computational efficiency compared to baselines. FLTrust achieves the lowest computation time but offers no privacy. Among Byzantine-robust PPFL schemes, our non-compression version significantly outperforms ShieldFL (which has similar privacy guarantees), confirming the effectiveness of our combined encryption-masking approach. The compression version achieves efficiency comparable to Split Aggregation while providing stronger security. Our compression reduces computational overhead by $25 \times$ to $35 \times$ compared to the non-compression version across different architectures, and substantially outperforms ShieldFL.

\begin{figure}[!t]
  \centering
  \subfloat[Computational overhead.]
  {\includegraphics[width=0.23\textwidth]{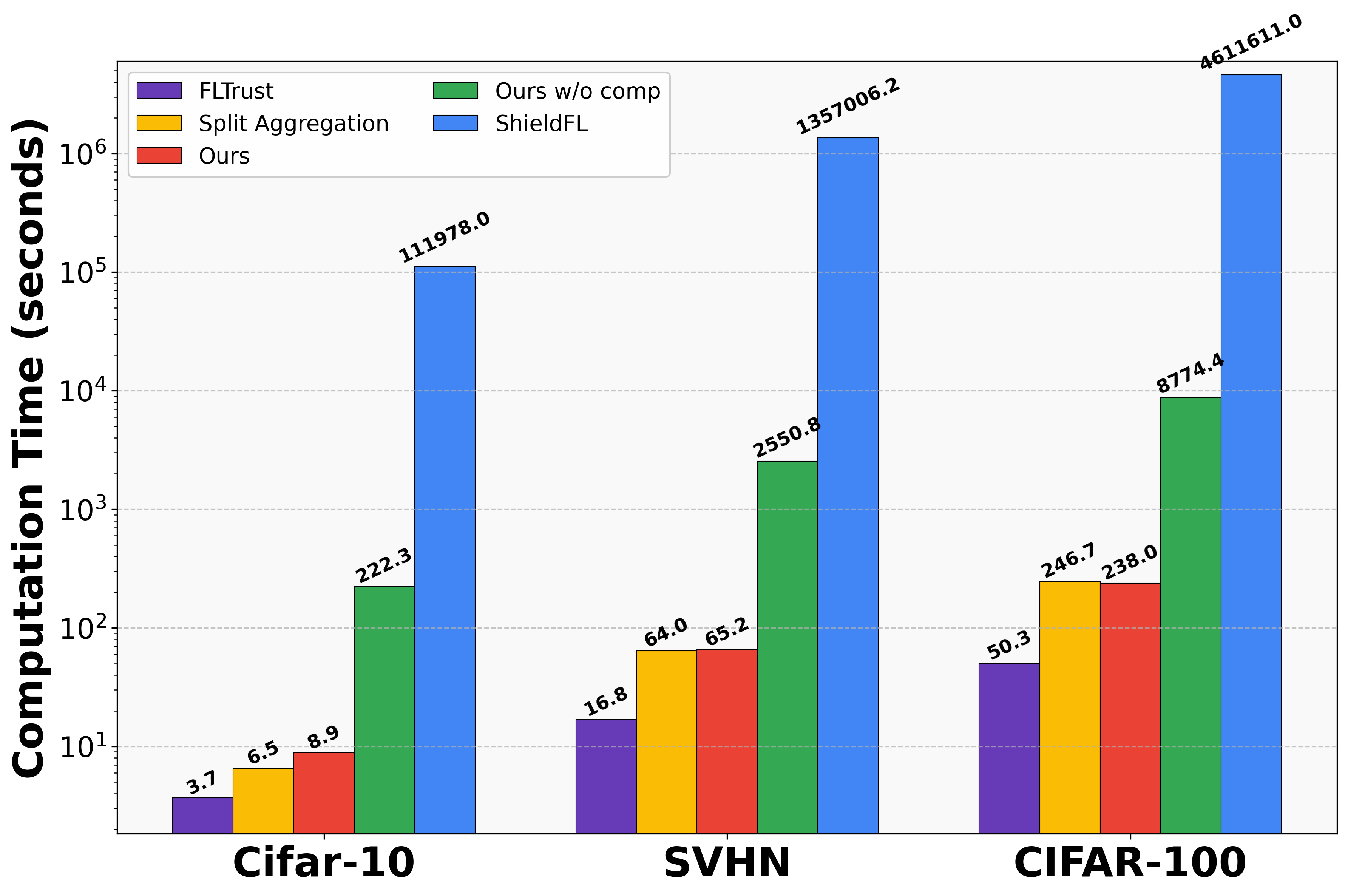}%
  \label{fig:efficiency}}
  \hfil
  \subfloat[Communication overhead.]
  {\includegraphics[width=0.23\textwidth]{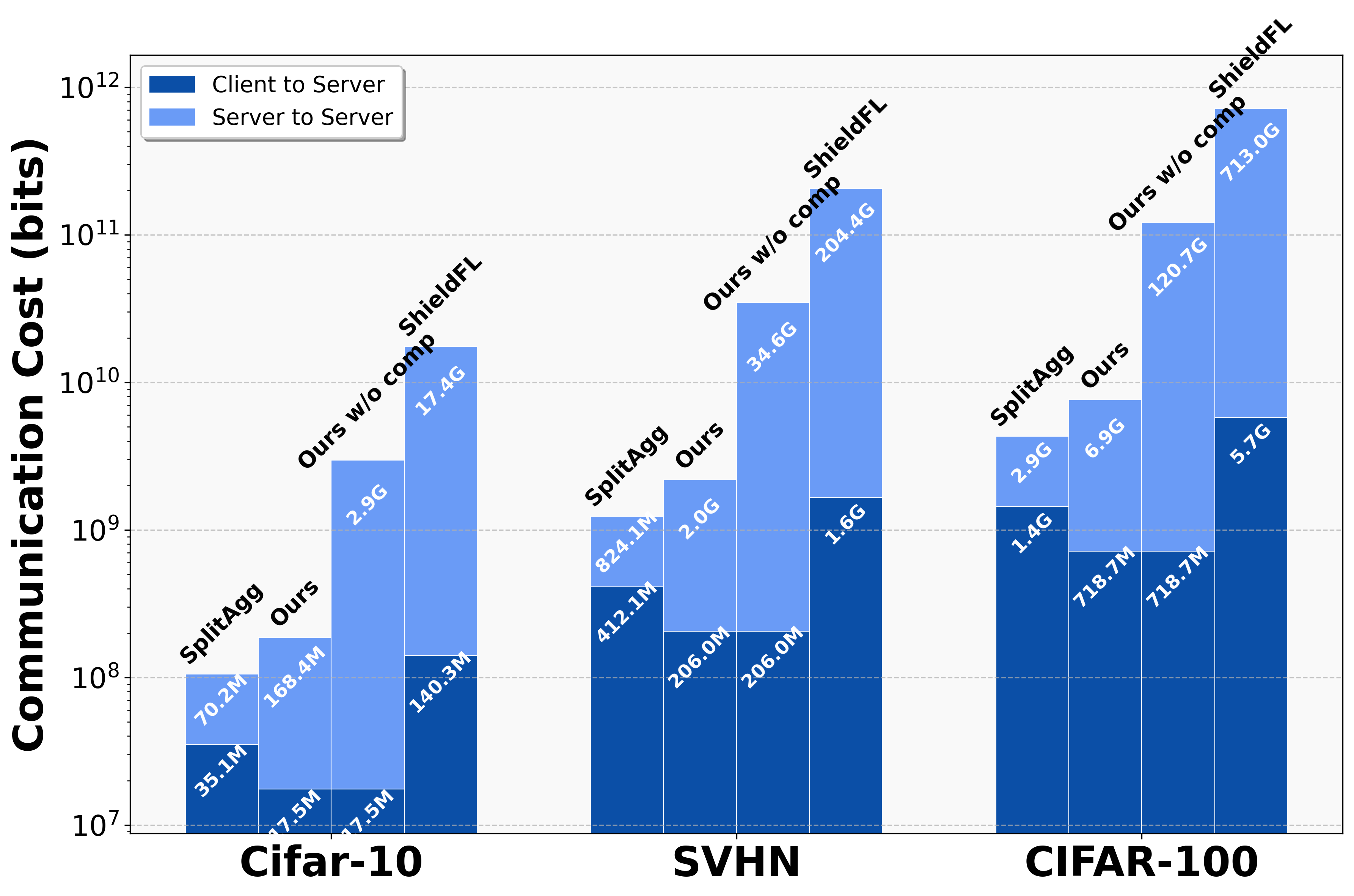}%
  \label{fig:comm}}
  \caption{Efficiency comparison showing (a) per-round processing time (seconds) with varying compression ratios (1.0 to 0.0001) across different model architectures under 40\% Byzantine client and (b) data transmission volume (bits) at different compression ratios across different datasets.}
  \label{fig:efficiency_comm}
\end{figure}

Figure \ref{fig:comm} illustrates our scheme's communication overhead compared to baselines. Our approach significantly reduces data transmission across all communication paths, requiring only half the overhead of SMC-based techniques like Split Aggregation and substantially less than ShieldFL for client-to-server communication. More importantly, our compression technique reduces server-to-server communication by approximately $17 \times$ compared to our non-compression version and orders of magnitude lower than ShieldFL.

\subsubsection{Impact of compression ratios on Accuracy and Efficiency}

\begin{figure}[!t]
  \centering
  \subfloat[Model accuracy]
  {\includegraphics[width=0.23\textwidth]{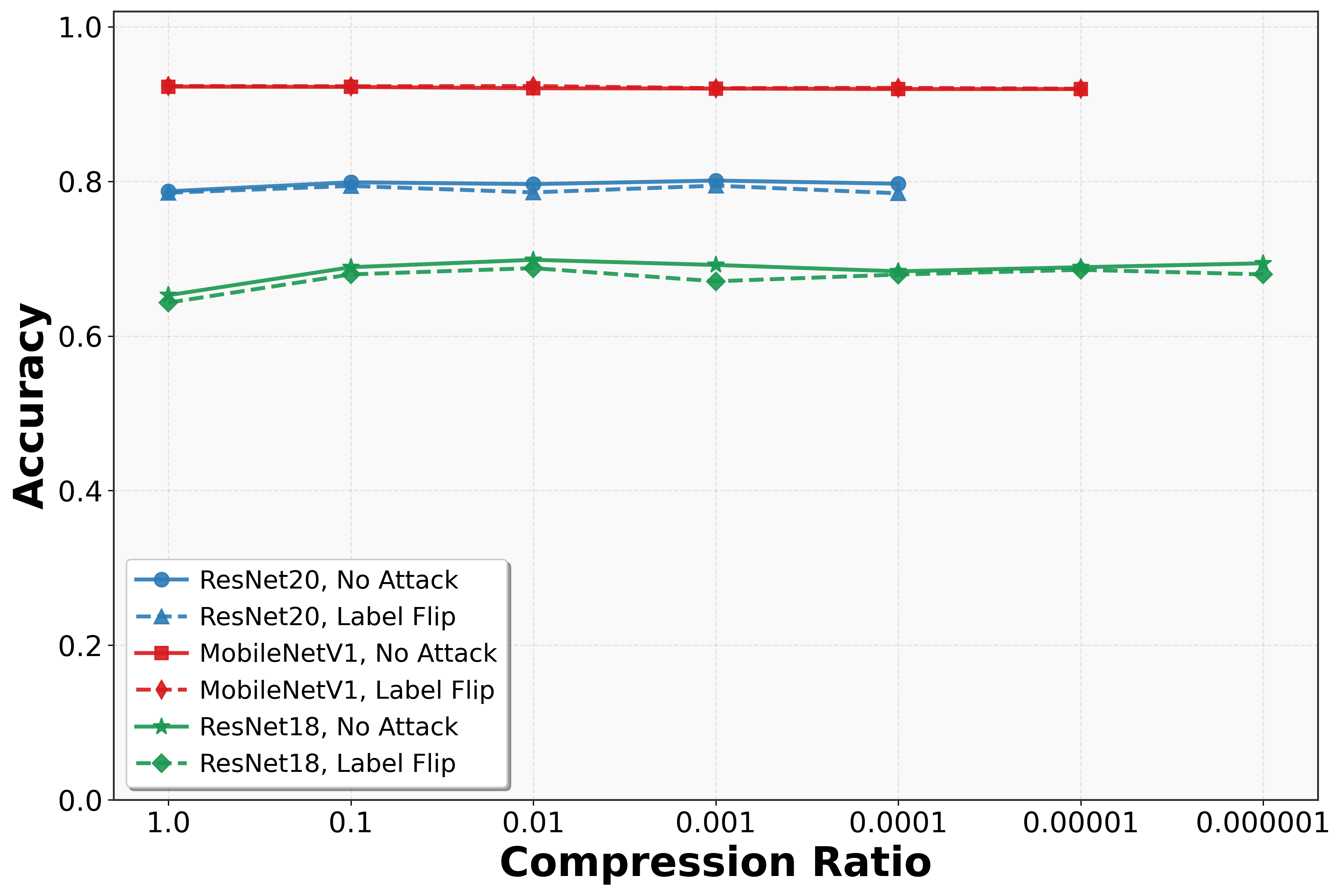}%
  \label{fig:cr_acc}}
  \hfil
  \subfloat[Processing time]
  {\includegraphics[width=0.23\textwidth]{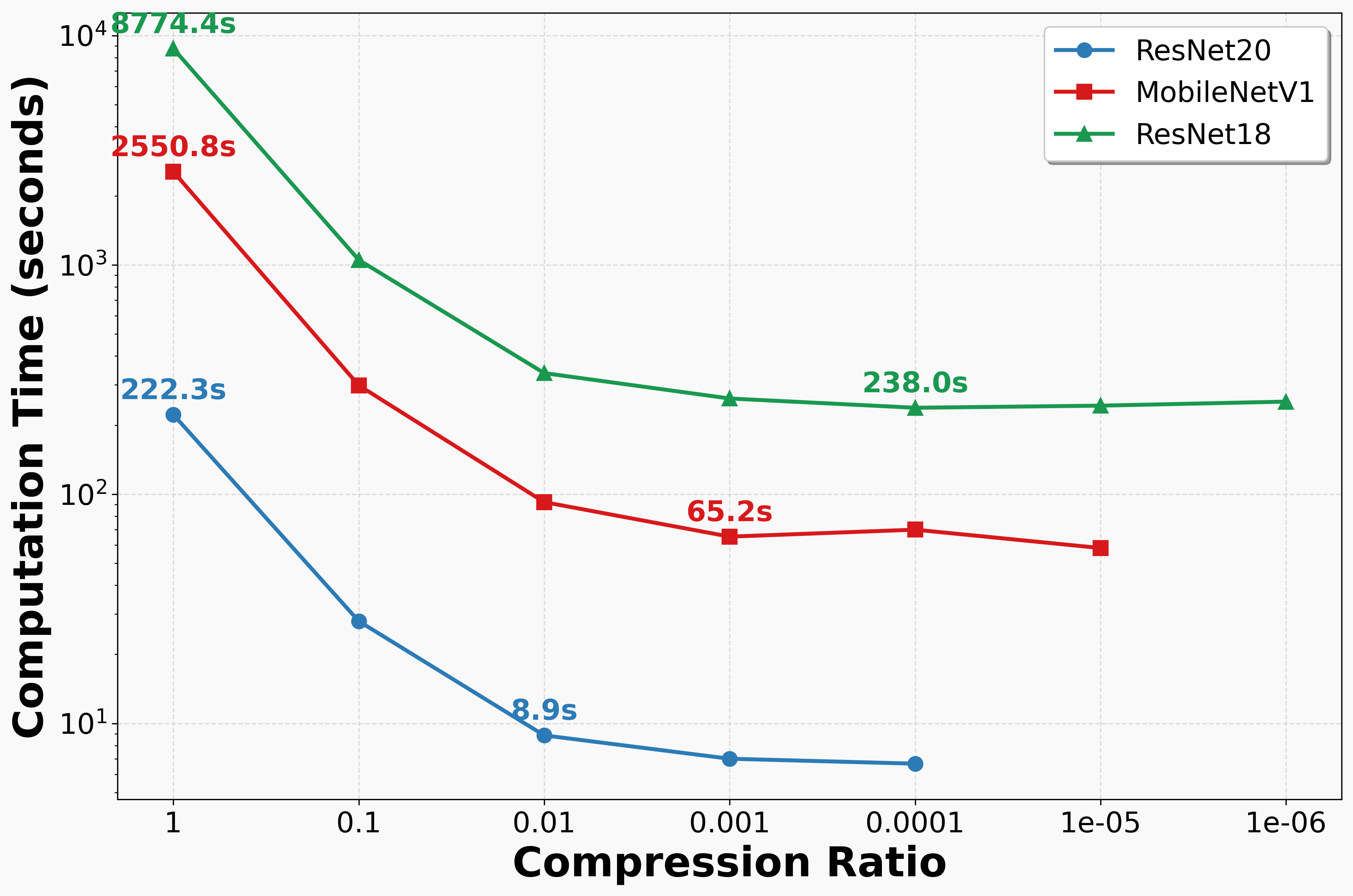}%
  \label{fig:cr_time}}
  \caption{Impact of compression ratios on model accuracy and computational efficiency under different compression ratios with 40\% Byzantine clients. (a) Model accuracy remains stable even at extreme compression (0.0001). (b) Processing time per round decreases dramatically with increasing compression.}
  \label{fig:compression_analysis}
\end{figure}

We evaluated our dimension compression technique across compression ratios from $1.0$ to $0.0001$ with $40\%$ Byzantine clients. Figure \ref{fig:cr_acc} shows that model accuracy remains stable across all datasets even at extreme compression, validating our theoretical analysis in Section \ref{sec:analysis_jl} that JL transformations preserve geometric relationships essential for Byzantine defense.

Figure \ref{fig:cr_time} demonstrates significant computational benefits, with per-round processing time decreasing dramatically as compression increases. At extreme ratios, our method achieves $25 \times$ to $35 \times$ reduction in computational overhead across different model architectures while maintaining defense capabilities comparable to uncompressed methods, making Byzantine-robust PPFL practical even for resource-constrained environments.

\section{Conclusion} \label{sec:conclusion}
We proposed a novel Byzantine-robust PPFL scheme that effectively balances privacy preservation, Byzantine robustness, and computational efficiency. Our approach combines additive masking with homomorphic encryption in a dual-server architecture, applying encryption only to Byzantine defense operations. To address computation bottlenecks, we introduced JL transformation-based dimension compression, reducing complexity from $O(dn)$ to $O(kn)$ where $k \ll d$.

Experimental results show our approach maintains model accuracy comparable to non-private FL while defending against up to 40\% Byzantine clients. The compression technique achieves $25 \times \sim 35 \times$ computational overhead reduction and $17 \times$ communication overhead reduction compared to our non-compression version, with order-of-magnitude improvements over ShieldFL, making Byzantine-robust PPFL practical for large-scale deployments.

\bibliographystyle{IEEEtran}
\bibliography{references}

\begin{thebibliography}{10}
\providecommand{\url}[1]{#1}
\csname url@samestyle\endcsname
\providecommand{\newblock}{\relax}
\providecommand{\bibinfo}[2]{#2}
\providecommand{\BIBentrySTDinterwordspacing}{\spaceskip=0pt\relax}
\providecommand{\BIBentryALTinterwordstretchfactor}{4}
\providecommand{\BIBentryALTinterwordspacing}{\spaceskip=\fontdimen2\font plus
\BIBentryALTinterwordstretchfactor\fontdimen3\font minus
  \fontdimen4\font\relax}
\providecommand{\BIBforeignlanguage}[2]{{%
\expandafter\ifx\csname l@#1\endcsname\relax
\typeout{** WARNING: IEEEtran.bst: No hyphenation pattern has been}%
\typeout{** loaded for the language `#1'. Using the pattern for}%
\typeout{** the default language instead.}%
\else
\language=\csname l@#1\endcsname
\fi
#2}}
\providecommand{\BIBdecl}{\relax}
\BIBdecl

\bibitem{ShieldFL}
Z.~Ma, J.~Ma, Y.~Miao, Y.~Li, and R.~H. Deng, ``Shieldfl: Mitigating model
  poisoning attacks in privacy-preserving federated learning,'' \emph{IEEE
  Transactions on Information Forensics and Security}, vol.~17, pp. 1639--1654,
  2022.

\bibitem{fltrust}
X.~Cao, M.~Fang, J.~Liu, and N.~Z. Gong, ``Fltrust: Byzantine-robust federated
  learning via trust bootstrapping,'' in \emph{ISOC Network and Distributed
  System Security Symposium (NDSS)}, 2021.

\bibitem{FedAvg}
B.~McMahan, E.~Moore, D.~Ramage, S.~Hampson, and B.~A.~y. Arcas,
  ``{Communication-Efficient Learning of Deep Networks from Decentralized
  Data},'' in \emph{Proceedings of the 20th International Conference on
  Artificial Intelligence and Statistics}, ser. Proceedings of Machine Learning
  Research, A.~Singh and J.~Zhu, Eds., vol.~54.\hskip 1em plus 0.5em minus
  0.4em\relax PMLR, 20--22 Apr 2017, pp. 1273--1282.

\bibitem{medicalsurvey}
\BIBentryALTinterwordspacing
H.~Guan, P.-T. Yap, A.~Bozoki, and M.~Liu, ``Federated learning for medical
  image analysis: A survey,'' \emph{Pattern Recognition}, vol. 151, p. 110424,
  2024. [Online]. Available:
  \url{https://www.sciencedirect.com/science/article/pii/S0031320324001754}
\BIBentrySTDinterwordspacing

\bibitem{FLSurvey-financial}
Y.~Shi, H.~Song, and J.~Xu, ``Responsible and effective federated learning in
  financial services: A comprehensive survey,'' in \emph{2023 62nd IEEE
  Conference on Decision and Control (CDC)}, 2023, pp. 4229--4236.

\bibitem{IOT}
D.~C. Nguyen, M.~Ding, P.~N. Pathirana, A.~Seneviratne, J.~Li, and
  H.~Vincent~Poor, ``Federated learning for internet of things: A comprehensive
  survey,'' \emph{IEEE Communications Surveys \& Tutorials}, vol.~23, no.~3,
  pp. 1622--1658, 2021.

\bibitem{DLG}
L.~Zhu, Z.~Liu, and S.~Han, \emph{Deep leakage from gradients}.\hskip 1em plus
  0.5em minus 0.4em\relax Red Hook, NY, USA: Curran Associates Inc., 2019.

\bibitem{privacyInference}
L.~Melis, C.~Song, E.~De~Cristofaro, and V.~Shmatikov, ``Exploiting unintended
  feature leakage in collaborative learning,'' in \emph{2019 IEEE Symposium on
  Security and Privacy (SP)}, 2019, pp. 691--706.

\bibitem{Membershipinference}
J.~Zhang, J.~Zhang, J.~Chen, and S.~Yu, ``Gan enhanced membership inference: A
  passive local attack in federated learning,'' in \emph{ICC 2020-2020 IEEE
  International Conference on Communications (ICC)}.\hskip 1em plus 0.5em minus
  0.4em\relax IEEE, 2020, pp. 1--6.

\bibitem{local_poisoning}
M.~Fang, X.~Cao, J.~Jia, and N.~Z. Gong, ``Local model poisoning attacks to
  byzantine-robust federated learning,'' in \emph{Proceedings of the 29th
  USENIX Conference on Security Symposium}, ser. SEC'20.\hskip 1em plus 0.5em
  minus 0.4em\relax USA: USENIX Association, 2020.

\bibitem{manipulating}
V.~Shejwalkar and A.~Houmansadr, ``Manipulating the byzantine: Optimizing model
  poisoning attacks and defenses for federated learning,'' in \emph{Network and
  Distributed System Security Symposium (NDSS)}, 2021.

\bibitem{EMKSA}
X.~Yang, Z.~Liu, X.~Tang, R.~Lu, and B.~Liu, ``An efficient and multi-private
  key secure aggregation scheme for federated learning,'' \emph{IEEE
  Transactions on Services Computing}, pp. 1--14, 2024.

\bibitem{LightSec}
J.~So, C.~He, C.-S. Yang, S.~Li, Q.~Yu, R.~E. Ali, B.~Guler, and S.~Avestimehr,
  ``Lightsecagg: a lightweight and versatile design for secure aggregation in
  federated learning,'' in \emph{Conference on Machine Learning and Systems},
  2021.

\bibitem{SAMFL}
\BIBentryALTinterwordspacing
M.~Guan, H.~Bao, Z.~Li, H.~Pan, C.~Huang, and H.-N. Dai, ``Samfl: Secure
  aggregation mechanism for federated learning with byzantine-robustness by
  functional encryption,'' \emph{Journal of Systems Architecture}, vol. 157, p.
  103304, 2024. [Online]. Available:
  \url{https://www.sciencedirect.com/science/article/pii/S1383762124002418}
\BIBentrySTDinterwordspacing

\bibitem{RFed}
Y.~Miao, X.~Yan, X.~Li, S.~Xu, X.~Liu, H.~Li, and R.~H. Deng, ``Rfed:
  Robustness-enhanced privacy-preserving federated learning against poisoning
  attack,'' \emph{IEEE Transactions on Information Forensics and Security},
  vol.~19, pp. 5814--5827, 2024.

\bibitem{LSFL}
Z.~Zhang, L.~Wu, C.~Ma, J.~Li, J.~Wang, Q.~Wang, and S.~Yu, ``Lsfl: A
  lightweight and secure federated learning scheme for edge computing,''
  \emph{IEEE Transactions on Information Forensics and Security}, vol.~18, pp.
  365--379, 2023.

\bibitem{SplitAgg}
Z.~Lu, S.~Lu, Y.~Cui, X.~Tang, and J.~Wu, ``Split aggregation: Lightweight
  privacy-preserving federated learning resistant to byzantine attacks,''
  \emph{IEEE Transactions on Information Forensics and Security}, vol.~19, pp.
  5575--5590, 2024.

\bibitem{trimmedmean}
\BIBentryALTinterwordspacing
D.~Yin, Y.~Chen, R.~Kannan, and P.~Bartlett, ``{B}yzantine-robust distributed
  learning: Towards optimal statistical rates,'' in \emph{Proceedings of the
  35th International Conference on Machine Learning}, ser. Proceedings of
  Machine Learning Research, J.~Dy and A.~Krause, Eds., vol.~80.\hskip 1em plus
  0.5em minus 0.4em\relax PMLR, 10--15 Jul 2018, pp. 5650--5659. [Online].
  Available: \url{https://proceedings.mlr.press/v80/yin18a.html}
\BIBentrySTDinterwordspacing

\bibitem{krum}
\BIBentryALTinterwordspacing
P.~Blanchard, E.~M. El~Mhamdi, R.~Guerraoui, and J.~Stainer, ``Machine learning
  with adversaries: Byzantine tolerant gradient descent,'' in \emph{Advances in
  Neural Information Processing Systems}, I.~Guyon, U.~V. Luxburg, S.~Bengio,
  H.~Wallach, R.~Fergus, S.~Vishwanathan, and R.~Garnett, Eds., vol.~30.\hskip
  1em plus 0.5em minus 0.4em\relax Curran Associates, Inc., 2017. [Online].
  Available:
  \url{https://proceedings.neurips.cc/paper_files/paper/2017/file/f4b9ec30ad9f68f89b29639786cb62ef-Paper.pdf}
\BIBentrySTDinterwordspacing

\bibitem{PEFL}
X.~Liu, H.~Li, G.~Xu, Z.~Chen, X.~Huang, and R.~Lu, ``Privacy-enhanced
  federated learning against poisoning adversaries,'' \emph{IEEE Transactions
  on Information Forensics and Security}, vol.~16, pp. 4574--4588, 2021.

\bibitem{DefendFL}
J.~Liu, X.~Li, X.~Liu, H.~Zhang, Y.~Miao, and R.~H. Deng, ``Defendfl: A
  privacy-preserving federated learning scheme against poisoning attacks,''
  \emph{IEEE Transactions on Neural Networks and Learning Systems}, pp. 1--14,
  2024.

\bibitem{Li2024}
X.~Li, X.~Yang, Z.~Zhou, and R.~Lu, ``Efficiently achieving privacy
  preservation and poisoning attack resistance in federated learning,''
  \emph{IEEE Transactions on Information Forensics and Security}, vol.~19, pp.
  4358--4373, 2024.

\bibitem{LSFL-erratum}
J.~Wu, W.~Zhang, and F.~Luo, ``On the security of “lsfl: A lightweight and
  secure federated learning scheme for edge computing”,'' \emph{IEEE
  Transactions on Information Forensics and Security}, vol.~19, pp. 3481--3482,
  2024.

\bibitem{PEFL-erratum}
T.~Schneider, A.~Suresh, and H.~Yalame, ``Comments on “privacy-enhanced
  federated learning against poisoning adversaries”,'' \emph{IEEE
  Transactions on Information Forensics and Security}, vol.~18, pp. 1407--1409,
  2023.

\bibitem{HE}
R.~L. Rivest and M.~L. Dertouzos, ``{ON DATA BANKS AND PRIVACY
  HOMOMORPHISMS},'' 1978.

\bibitem{HEsurvey}
\BIBentryALTinterwordspacing
M.~Albrecht, M.~Chase, H.~Chen, J.~Ding, S.~Goldwasser, S.~Gorbunov, S.~Halevi,
  J.~Hoffstein, K.~Laine, K.~Lauter, S.~Lokam, D.~Micciancio, D.~Moody,
  T.~Morrison, A.~Sahai, and V.~Vaikuntanathan, ``Homomorphic encryption
  standard,'' Cryptology {ePrint} Archive, Paper 2019/939, 2019. [Online].
  Available: \url{https://eprint.iacr.org/2019/939}
\BIBentrySTDinterwordspacing

\bibitem{Paillier}
P.~Paillier, ``Public-key cryptosystems based on composite degree residuosity
  classes,'' in \emph{Advances in Cryptology --- EUROCRYPT '99}, J.~Stern,
  Ed.\hskip 1em plus 0.5em minus 0.4em\relax Berlin, Heidelberg: Springer
  Berlin Heidelberg, 1999, pp. 223--238.

\bibitem{JL-lemma}
\BIBentryALTinterwordspacing
J.~Fedoruk, B.~Schmuland, J.~Johnson, and G.~Heo, ``Dimensionality reduction
  via the johnson---lindenstrauss lemma: theoretical and empirical bounds on
  embedding dimension,'' \emph{J. Supercomput.}, vol.~74, no.~8, p.
  3933–3949, Aug. 2018. [Online]. Available:
  \url{https://doi.org/10.1007/s11227-018-2401-y}
\BIBentrySTDinterwordspacing

\bibitem{howtosimulate}
\BIBentryALTinterwordspacing
Y.~Lindell, ``How to simulate it,'' in \emph{Advances in Cryptology - EUROCRYPT
  2017}.\hskip 1em plus 0.5em minus 0.4em\relax Springer, 2017, pp. 1--25.
  [Online]. Available:
  \url{https://link.springer.com/chapter/10.1007/978-3-662-46803-6_18}
\BIBentrySTDinterwordspacing

\bibitem{paszke2019pytorch}
A.~Paszke, S.~Gross, F.~Massa, A.~Lerer, J.~Bradbury, G.~Chanan, T.~Killeen,
  Z.~Lin, N.~Gimelshein, L.~Antiga \emph{et~al.}, ``Pytorch: An imperative
  style, high-performance deep learning library,'' in \emph{Advances in Neural
  Information Processing Systems (NeurIPS)}, vol.~32.\hskip 1em plus 0.5em
  minus 0.4em\relax Curran Associates, Inc., 2019, pp. 8024--8035.

\bibitem{svhn}
\BIBentryALTinterwordspacing
Y.~Netzer, T.~Wang, A.~Coates, A.~Bissacco, B.~Wu, and A.~Y. Ng, ``Reading
  digits in natural images with unsupervised feature learning,'' in \emph{NIPS
  Workshop on Deep Learning and Unsupervised Feature Learning 2011}, 2011.
  [Online]. Available:
  \url{http://ufldl.stanford.edu/housenumbers/nips2011_housenumbers.pdf}
\BIBentrySTDinterwordspacing

\bibitem{cifar}
A.~Krizhevsky, I.~Sutskever, and G.~E. Hinton, ``Imagenet classification with
  deep convolutional neural networks,'' \emph{Commun. ACM}, vol.~60, no.~6, p.
  84–90, may 2017.

\bibitem{mobilenet}
A.~G. Howard, M.~Zhu, B.~Chen, D.~Kalenichenko, W.~Wang, T.~Weyand,
  M.~Andreetto, and H.~Adam, ``{MobileNets: Efficient Convolutional Neural
  Networks for Mobile Vision Applications},'' 4 2017.

\end{thebibliography}

\end{document}